\documentclass[11pt,a4paper,reqno]{amsart}%
\usepackage{amsthm,amsmath,amsfonts,amssymb,amsxtra,appendix,bookmark,color,dsfont,bm,mathrsfs}
\usepackage[margin=33mm]{geometry}
\usepackage[english]{babel}
\usepackage{enumerate}

\theoremstyle{plain}
\newtheorem{theorem}{Theorem}
\newtheorem{lemma}[theorem]{Lemma}
\newtheorem{corollary}[theorem]{Corollary}
\newtheorem{proposition}[theorem]{Proposition}

\theoremstyle{definition}

\theoremstyle{remark}

\def\leqslant{\le}
\def\bq{\begin{eqnarray}}
\def\eq{\end{eqnarray}}
\def\bqq{\begin{align*}}
\def\eqq{\end{align*}}

\def\eps{\varepsilon}

\renewcommand{\epsilon}{\varepsilon}
\newcommand\1{{\ensuremath {\mathds 1} }}

\def\cF {\mathcal{F}}

\def\cH{\mathcal{H}}
\def\cN{\mathcal{N}}
\def\cV {\mathcal{V}}
\def\cL {\mathcal{L}}
\def\R {\mathbb{R}}

\def\N {\mathcal{N}}

\def\cG {\mathcal{G}}

\def\cE {\mathcal{E}}
\def\cA{\mathcal{A}}

\def\cK{\mathcal{K}}

\def\cV {\mathcal{V}}
\def\R {\mathbb{R}}

\def\N {\mathcal{N}}

\def\cA{\mathcal{A}}

\def\d{\,{\rm d}}
\def\w{\widetilde}
\def\cosh{{\rm cosh}}
\def\sinh{{\rm sinh}}

\newcommand{\dGamma}{{\ensuremath{\rm d}\Gamma}}

\def\bR{\mathbb{R}}
\def\bN{\mathbb{N}}

\def\cU{\mathcal{U}}

\def\cV{\mathcal{V}}
\def\cF{\mathcal{F}}
\def\cG{\mathcal{G}}

\def\cL{\mathcal{L}}
\def\cN{\mathcal{N}}
\def\cR{\mathcal{R}}
\def\cE{\mathcal{E}}
\def\cK{\mathcal{K}}
\def\cH{\mathcal{H}}
\def\eps{\varepsilon}
\def\ph{\varphi}

\def\wt{\widetilde}

\def\indic{\hbox{\raise-2pt \hbox{\indbf 1}}}

\let\==\equiv

\let\0=\noindent

\def\*{{\hfill\break\null\hfill\break}}

\def\tende#1{\,\vtop{\ialign{##\crcr\rightarrowfill\crcr
             \noalign{\kern-1pt\nointerlineskip}
             \hskip3.pt${\scriptstyle #1}$\hskip3.pt\crcr}}\,}
\def\otto{\,{\kern-1.truept\leftarrow\kern-5.truept\to\kern-1.truept}\,}

\def\tr{{\rm tr}}

\def\Re{{\rm Re}\,}\def\Im{{\rm Im}\,}

\newcommand{\product}[2]{\ensuremath{\left\langle #1, #2 \right \rangle}}

\newcommand{\im}[0]{\ensuremath{\operatorname{Im}}}

\allowdisplaybreaks

\title[Quantum fluctuations]{Fluctuations of $N$-particle quantum dynamics around the nonlinear Schr\"odinger equation}

\author[C. Brennecke]{Christian Brennecke}
\address{Institute of Mathematics, University of Zurich, 
Winterthurerstrasse 190, 8057 Zurich, Switzerland} 
\email{christian.brennecke@math.uzh.ch}
 
\author[P.T. Nam]{Phan Th\`anh Nam}
\address{Department of Mathematics, LMU Munich, Theresienstrasse 39, 80333 Munich, Germany} 
\email{nam@math.lmu.de}

\author[M. Napi\'orkowski]{Marcin Napi\'orkowski}
\address{Department of Mathematical Methods in Physics, Faculty of Physics, University of Warsaw,  Pasteura 5, 02-093 Warszawa, Poland}
\email{marcin.napiorkowski@fuw.edu.pl} 

\author[B. Schlein]{Benjamin Schlein}
\address{Institute of Mathematics, University of Zurich, 
Winterthurerstrasse 190, 8057 Zurich, Switzerland} 
\email{benjamin.schlein@math.uzh.ch}

\begin{document}
\date{\today}

\begin{abstract} 
We consider a system of $N$ bosons interacting through a singular  two-body potential scaling with $N$ and having the  form 
$N^{3\beta-1} V (N^\beta x)$, for an arbitrary parameter $\beta \in (0,1)$. We provide 
a norm-approximation for the many-body evolution of initial data exhibiting Bose-Einstein condensation in terms of a cubic 
nonlinear Schr\"odinger equation for the condensate wave function and of a unitary Fock space evolution with a generator quadratic in creation and annihilation operators for the fluctuations. 
\end{abstract}

\maketitle


\section{Introduction}

From first principles of quantum mechanics, the evolution of a system of $N$ identical (spinless) 
bosons in $\R^3$ is governed by the many-body Schr\"odinger equation 
\begin{equation}\label{eq:schr} i\partial_t \Psi_{N,t} = H_N \Psi_{N,t} \end{equation}
where 
$$\Psi_{N,t} \in L_{s}^2(\R^{3N}) = L^2(\R^3)^{\otimes_s N}$$
is the wave function and $H_N$ is the Hamilton operator of the system. We will restrict our attention to Hamilton operators of the form   
\begin{equation} \label{eq:HN}
H_N= \sum\limits_{j = 1}^N -\Delta_{x_j} + \frac{1}{N} \sum\limits_{1 \leqslant j < k \leqslant N} {V_N(x_j-x_k)}
\end{equation}
with $N$-dependent two-body interaction potential 
\begin{equation}  \label{eq:ass-wN}
V_N(x)= N^{3\beta} V(N^\beta x).
\end{equation}
Here $\beta \geq 0$ is a fixed parameter and $V \ge 0$ is a smooth, radially symmetric and compactly supported function on $\bR^3$. 

For $\beta = 0$, (\ref{eq:HN}) is a mean-field Hamiltonian, describing a system of particles experiencing a large number of weak collisions. For $\beta = 1$, on the other hand, (\ref{eq:HN}) corresponds to the Gross-Pitaevskii regime, where collisions are rare but strong. Physically, the Gross-Pitaevskii regime is more relevant for the description of trapped Bose-Einstein condensates. The mean-field regime, on the other hand, is more accessible to mathematical analysis. In this paper, we will study the solution of the Schr\"odinger equation (\ref{eq:schr}) for intermediate regimes with $0 < \beta  < 1$. 

{F}rom the point of view of physics, it is interesting to study the solution of (\ref{eq:schr}) for initial data approximating ground states of trapped systems; this corresponds to experimental settings where the evolution of an initially trapped Bose gas at very low temperature is observed after switching off the external fields. 

It is known since \cite{LS,NRS} that the ground state of a system of trapped bosons interacting through a two-body potential like the one appearing on the r.h.s. of (\ref{eq:HN}) exhibits complete Bose-Einstein condensation (BEC); the one-particle reduced density associated with the ground state wave function $\psi_N \in L^2_s (\bR^{3N})$ converges, as $N \to \infty$, towards the orthogonal projection onto a one-particle orbital $\ph_0 \in L^2 (\bR^3)$. 

Hence, we will be interested in the solution of (\ref{eq:schr}) for initial data 
exhibiting BEC. Despite its linearity, for large $N$ ($N \simeq 10^5 - 10^7$ in typical experiments) it 
is impossible to solve the many-body Schr\"odinger equation (\ref{eq:schr}), neither analytically nor  numerically. It is important, therefore, to find good approximations of 
the solution of (\ref{eq:schr}) that are valid in the limit $N \to \infty$. A first step in this direction was achieved in \cite{ESY1} for $\beta < 1/2$ and in \cite{ESY2,ESY3} for the Gross-Pitaevskii regime with $\beta = 1$ (the same ideas can also be extended to all $\beta \in (0,1)$), where it was proven that, for every fixed time $t \in \bR$, the solution $\psi_{N,t}$ of (\ref{eq:schr}) still exhibits BEC and that its one-particle reduced density converges to the orthogonal projection onto $\ph_t$, given by the solution of the cubic nonlinear Schr\"odinger equation 
\begin{equation}\label{eq:NLS0} i\partial_t \ph_t = - \Delta \ph_t + \sigma |\ph_t|^2 \ph_t \end{equation}
with the initial data $\ph_{t=0} = \ph$ and with coupling constant $\sigma = \int V(x) dx$ for $\beta < 1$ and $\sigma = 8\pi a_0$ for $\beta =1$ (where $a_0$ denotes the scattering length of the unscaled potential $V$).  
The results of \cite{ESY1,ESY2,ESY3} have been revisited and improved further in \cite{P,BDS,CH2,BS}. 
In the simpler case $\beta = 0$, i.e. in the mean-field regime, the convergence of the one-particle reduced density towards the orthogonal projection onto the solution of the nonlinear Hartree equation
\begin{equation}\label{eq:hartree}
i\partial_t \ph_t = -\Delta \ph_t + (V * |\ph_t|^2) \ph_t 
\end{equation}
has been proved in several situations; see, e.g., \cite{Spohn,BGM,EY,AGT,ES,AN1,FKP,FKS,KP,KSS,AH,CH,AFP}.  


In the present paper, we are interested in the norm approximation to the many-body evolution, which is more precise than the convergence of the one-particle reduced density. It requires not only to follow the dynamics of the condensate, but also to take into account the evolution of its excitations.

To describe excitations and their dynamics, it is convenient to switch to a Fock space representation (because the number of excitations, in contrast with the total number of particles, is not preserved). We define the bosonic Fock space 
\[ \cF = \bigoplus_{n \geq 0} L^2_s (\bR^{3n}). \]
For $f \in L^2 (\bR^3)$ and for $\Psi \in \cF$, we define the creation operator $a^* (f)$ and its adjoint, the annihilation operator $a(f)$, through 
\begin{align*}
(a^* (f) \Psi )^{(n)} (x_1,\dots,x_{n})&= \frac{1}{\sqrt{n}} \sum_{j=1}^{n} f(x_j) \Psi^{(n-1)} (x_1,\dots,x_{j-1},x_{j+1},\dots, x_{n+1}), \\
(a(f) \Psi)^{(n)} (x_1,\dots,x_{n}) &= \sqrt{n+1} \int \overline{f(x_n)} \Psi^{(n+1)} (x_1,\dots,x_n, x_{n+1}) \d x_{n+1}
\end{align*} 
Creation and annihilation operators satisfy canonical commutation relations (CCR)
\begin{equation}\label{eq:CCR}
[a(f),a(g)]=[a^*(f),a^*(g)]=0,\quad [a(f), a^* (g)]= \langle f, g \rangle, \quad \forall f,g \in L^2(\R^3).
\end{equation}
It is also convenient to introduce operator-valued distributions $a_x^*$ and $a_x$ so that 
\begin{equation}\label{eq:opval}
a^*(f)=\int_{\R^3}  f(x) a_x^* \d x, \quad a(f)=\int_{\R^3} \overline{f(x)} a_x \d x, \quad \forall f\in  L^2(\R^3).
\end{equation}
Expressed through these operator-valued distributions, the CCR take the form 
$$[a^*_x,a^*_y]=[a_x,a_y]=0, \quad [a_x,a^*_y]=\delta(x-y), \quad \forall x,y\in \R^3.$$
A self-adjoint operator $A$ on the one-particle space $L^2 (\bR^3)$ can be lifted to a Fock space operator by second quantization, defining 
\[ d\Gamma (A) = \bigoplus_{n=0}^\infty \sum_{j=0}^n A_{j} \]
with $A_j$ acting as $A$ on the $j$-th particle and as the identity on the other $(N-1)$ particles. 
If $A$ has the integral kernel $A(x;y)$, $d\Gamma (A)$ can be expressed  as 
\[ d\Gamma (A) = \int A(x;y) a_x^* a_y \, dx dy \]
For example, the number of particles operator is given by 
\[ \cN = d\Gamma ( 1) = \int dx a_x^* a_x \]

On the Fock space $\cF$, it is instructive to study the time-evolution of coherent initial data, having the form 
\begin{equation}\label{eq:coh} W(\sqrt{N} \ph) \Omega = e^{-N/2} \left\{ 1, \ph , \frac{\ph^{\otimes 2}}{\sqrt{2!}} , \dots \right\} \end{equation} 
for $\ph \in L^2 (\bR^3)$ with $\| \ph \| = 1$. Here $\Omega = \{ 1, 0 ,0, \dots \}$ is the Fock space vacuum and, for any $f \in L^2 (\bR^3)$, $W(f) = \exp (a^* (f) - a(f))$ is a Weyl operator. The normalization of $\ph$ guarantees that
\[ \langle W(\sqrt{N} \ph) \Omega, \cN  W(\sqrt{N} \ph) \Omega \rangle = N. \]

The time-evolution of initial coherent states of the form (\ref{eq:coh}), generated by the natural extension of the Hamiltonian (\ref{eq:HN}) to the Fock space $\cF$
\begin{equation}\label{eq:Hbeta} \cH_N = \int dx a_x^* (-\Delta_x) a_x + \frac{1}{2N} \int \d x \d y \, V_N (x-y) a_x^* a_y^* a_y a_x =: \cK + \cV_N \end{equation}
has been studied for $\beta = 0$ in  \cite{Hepp,GV}, where it was proven that 
\begin{equation}\label{eq:Hepp} \left\| e^{-i\cH_N t} W(\sqrt{N} \ph) \Omega - W(\sqrt{N} \ph_t)  \, \cU_{2,\text{mf}}^\text{f} \, (t;0) \Omega \right\|  \to 0 \end{equation}
as $N \to \infty$. Here $\ph_t$ denotes the solution of the Hartree equation (\ref{eq:hartree}) and $\cU_{2,\text{mf}}^f (t;s)$ is a unitary dynamics on $\cF$ with a time-dependent generator that is quadratic in creation and annihilation operators \footnote{In the notation for $\cU_{2,\text{mf}}^\text{f}$, the subscript $\text{mf}$ and the superscript $\text{f}$ refer to the fact that (\ref{eq:Hepp}) holds in the mean-field regime with $\beta = 0$ for Fock space initial data}. This implies that $\cU^\text{f}_{2,\text{mf}} (t;s)$ acts on creation and annihilation operators as a time-dependent Bogoliubov transformation $\Theta_\text{mf} (t;s) : L^2 (\bR^3) \oplus L^2 (\bR^3) \to L^2 (\bR^3) \oplus L^2 (\bR^3)$ having the form
\begin{equation}\label{eq:Theta} \Theta_\text{mf} (t;s) = \left( \begin{array}{ll} U_\text{mf} (t;s) & \overline{V_\text{mf} (t;s)} \\ V_\text{mf} (t;s) & \overline{U_\text{mf} (t;s)} \end{array} \right). \end{equation}
In other words, for any $f \in L^2 (\bR^3)$ and all $t,s \in \bR$, we find 
\begin{equation}\label{eq:U2-action} \cU_{2,\text{mf}}^\text{f} (t;s)^* a(f) \, \cU_{2,\text{mf}}^\text{f} (t;s) = a (U_\text{mf} (t;s) f) + a^* (V_\text{mf} (t;s) \bar{f}). \end{equation}
The time-dependent Bogoliubov transformation $\Theta_\text{mf}$ can be determined solving the partial differential equation 
\begin{equation}\label{eq:dttheta} i \partial_t \Theta_\text{mf} (t;s) = \cA_\text{mf} (t) \Theta_\text{mf} (t;s) \end{equation}
with initial condition $\Theta_\text{mf} (s;s) = 1$ and with generator   
\[ \cA_\text{mf} (t) = \left( \begin{array}{ll} D (t) & - \overline{B (t)} \\ B (t) & - \overline{D (t)} \end{array} \right) \]
where
\[ \begin{split} D(t) f &= -\Delta f + (V*|\ph_t|^2) f + (V*\overline{\ph}_t f) \ph_t  \\
B(t) f &= (V*\overline{\ph}_t f) \overline{\ph}_t. \end{split} \]
Thus, (\ref{eq:Hepp}) allows us to describe the very complex many-body dynamics generated on $\cF$ by the Hamiltonian (\ref{eq:Hbeta}) by solving the equation (\ref{eq:hartree}) for the condensate wave function and the equation (\ref{eq:dttheta}) for the Bogoliubov transformation $\Theta_\text{mf} (t;s)$ describing the evolution of fluctuations around the condensate. 

The ideas of \cite{Hepp,GV} have been further developed in \cite{RS} and they have been used to prove a central limit theorem in \cite{BKS,BSS}. In \cite{GMM1,GMM2}, norm approximations for the many-body dynamics in Fock space has been derived using different approaches. 

To obtain a norm approximation for the mean-field time-evolution of $N$-particle initial data exhibiting BEC in a state with wave function $\ph \in L^2 (\bR^3)$, it is very convenient to use a unitary map introduced in \cite{LNSS}, mapping $L^2_s (\bR^{3N})$ into the 
truncated Fock space 
\begin{equation}\label{eq:trunc-F} \cF_{\perp \ph}^{\leq N} = \bigoplus_{j=0}^N L^2_{\perp \ph} (\bR^3)^{\otimes_s N} \end{equation}
constructed over the orthogonal complement $L^2_{\perp \ph} (\bR^3)$ of the one-dimensional space spanned by the condensate wave function $\ph$. The space (\ref{eq:trunc-F}) provides the natural setting to describe orthogonal excitations of the condensate (whose number can fluctuate). The idea here is that every $\psi_N \in L^2_s (\bR^{3N})$ can be written uniquely as 
\[ \psi_N = \alpha_0 \ph^{\otimes N} + \alpha_1 \otimes_s \ph^{\otimes (N-1)} + \dots + \alpha_N \]
where $\alpha_j \in L^2_{\perp \ph} (\bR^3)^{\otimes_s j}$ for all $j=0,\dots , N$ (for $j=0$, $\alpha_0 \in \mathbb{C}$). Therefore, we can define $U_\ph : L^2_s (\bR^{3N}) \to \cF_{\perp \ph}^{\leq N}$ by setting $U_\ph \psi_N = \{ \alpha_0, \dots , \alpha_N \}$. By orthogonality, it is easy to check that $U_\ph$ is a unitary map. In terms of creation and annihilation operators, it is given by 
\begin{align}\label{eq:U-def}
 U_{\varphi} = \bigoplus_{n=0}^N (1-|\varphi\rangle \langle \varphi|)^{\otimes n} \frac{a (\varphi)^{N-n}}{\sqrt{(N-n)!}}, \quad  U_{\varphi}^* = \sum_{n=0}^N \frac{a^* (\varphi)^{N-n}}{\sqrt{(N-n)!}}.
\end{align}
The actions of $U_\varphi$ on creation and annihilation operators follow the simple rules:
 \begin{align}\label{eq:U-rules}
    U_{\varphi} a^* (\varphi) a (\varphi) U_{\varphi}^*  & = N-\cN, \\
    U_{\varphi} a^*(f)a (\varphi) U_{\varphi}^* &= a^*(f) \sqrt{N-\cN},\\
    U_{\varphi} a^* (\varphi) a(g) U_{\varphi}^* &= \sqrt{N-\cN} a(g),\\
    U_{\varphi} a^* (f) a(g)U_{\varphi}^* &= a^*(f)a(g)
    \end{align}
for all $ f,g\in L^2_{\bot \varphi} (\bR^3)$. Heuristically, $U_{\ph}$ factors out the condensate described by the wave function $\ph$ and it allows us to focus on its orthogonal excitations. 

The unitary map $U_\ph$ was used in \cite{LNS} to obtain a norm approximation for the many-body evolution in the mean-field regime with $\beta =0$ (see \cite{MPP} for a similar result). For $N$-particle initial data of the form $\psi_N = U^*_\ph \xi_N$ with $\xi_N \in \cF_{\perp \ph}^{\leq N}$ having a finite expectation for the number of particles and for the kinetic energy operator, it was proven there that the solution of the many-body Schr\"odinger equation (\ref{eq:schr}) is such that  
\begin{equation}\label{eq:mf-LNS} \left\|  U_{\ph_t} \psi_{N,t} - \cU_{2,\text{mf}} (t;0) \xi_N  
\right\|  \to 0 \end{equation}
as $N \to \infty$, where, similarly to (\ref{eq:Hepp}), $\ph_t$ is the solution of (\ref{eq:hartree}) and $\cU_{2,\text{mf}} (t;s)$ is a unitary evolution on the Fock space, with a time-dependent generator quadratic in creation and annihilation operators (in fact $\cU_{2,\text{mf}}$ is very similar to the unitary evolution $\cU_{2,\text{mf}}^\text{f}$ in (\ref{eq:Hepp}), emerging in the mean field limit for coherent initial data on the Fock space). Eq. (\ref{eq:mf-LNS}) is the analogous of (\ref{eq:Hepp}) for $N$-particle initial data exhibiting BEC; it provides a norm-approximation of the many-body evolution in the mean-field regime in terms of the Hartree equation (\ref{eq:hartree}) and of a time-dependent Bogoliubov transformation very similar to (\ref{eq:Theta}). 

The convergence (\ref{eq:mf-LNS}) has been extended to intermediate regimes with $\beta<1/3$ in \cite{NN1} and with $\beta < 1/2$ in \cite{NN2}. Before that, a norm approximation similar to (\ref{eq:Hepp}) for the evolution of coherent initial data on the Fock space  has been obtained with $\beta<1/3$ in \cite{GM1} and with $\beta  <1/2$ in  \cite{K}. A heuristic argument from \cite{K} also shows that (\ref{eq:Hepp}) or (\ref{eq:mf-LNS}) cannot hold true for $\beta > 1/2$.

 In regimes with $\beta > 1/2$ the short scale correlation structure developed by the solution of the many-body Schr\"odinger equation cannot be appropriately described by a time-dependent Bogoliubov transformation satisfying an equation of the form (\ref{eq:dttheta}). To take into account correlations, it is useful to consider the ground state of the Neumann problem  
\begin{equation}\label{eq:Neum} \left[ -\Delta + \frac{1}{2N} V_N \right] f_{N}  = \lambda_{N} f_{N} \end{equation}
on the ball $|x| \leq \ell$, for a fixed $\ell > 0$. We fix $f_{N} (x) = 1$, for $|x| = \ell$, and we extend $f_{N}$ to $\bR^3$ requiring that $f_{N} (x) = 1$ for all $|x| \geq \ell$. Because of the scaling of the potential $V_N$, the scattering process takes place in the region $|x| \ll 1$; for this reason, the precise choice of $\ell$ is not very important, as long as $\ell$ is of order one (nevertheless, $\lambda_N$ and $f_N$ depend on $\ell$, a dependence that is kept implicit in our notation). It is also useful to define $\omega_{N} = 1 - f_N$. For $N$ sufficiently large, we have 
(see \cite[Lemma 2.1]{BCS})
\[ \lambda_N = \frac{3b_0}{8\pi N\ell^3} + O(N^{\beta-2}) \]
where $b_0=\int V(x)\,dx$, and, for all $x\in \R^3$,
\begin{equation} \label{eq:fN-bound}
0\le \omega_N(x) \le \frac{C}{N(|x|+N^{-\beta})},\;\hspace{0.5cm} |\nabla \omega_{N}(x)| \leq \frac{C}{N(|x|+N^{-\beta})^2} 
\end{equation}
for a constant $C$, independent of $N$. 

The solution of (\ref{eq:Neum}) can be used, first of all, to give a better approximation of the evolution of the condensate wave function, replacing the solution of the limiting nonlinear Schr\"odinger equation (\ref{eq:NLS0}) with the solution of the modified, $N$-dependent, Hartree equation 
\begin{equation}\label{eq:NLSN} 
i\partial \varphi_{N,t} = -\Delta \varphi_{N,t} + (V_N f_{N} *|\varphi_{N,t}|^2 ) \varphi_{N,t}
\end{equation}
with initial data $\varphi_{N,0} = \ph_0$ describing the condensate at time $t=0$. Standard 
arguments in the analysis of dispersive partial differential equations imply that (\ref{eq:NLSN}) is globally well-posed and that it propagates regularity; in particular, if $\varphi_0 \in H^4(\R^3)$, 
then \cite[Appendix B]{BCS}
\begin{equation}\label{eq:phi-bds} 
\| \varphi_{N,t} \|_{H^1} \le C, \quad \| \varphi_{N,t} \|_{H^4} \leq C e^{Ct}, \quad \|\partial_t \varphi_{N,t}\|_{H^2} \le C e^{Ct},  \quad \forall t>0.
\end{equation}

Furthermore, (\ref{eq:Neum}) can be used to describe correlations among particles. To this end, let 
\begin{equation}\label{eq:Bog-trans} T_{N,t} = \exp \left( \frac{1}{2} \int dx dy \, \left[  k_{N,t}(x,y) a_x a_y - \text{h.c.}  \right]  \right) \end{equation}
with the integral kernel
\begin{equation}\label{eq:kNt} 
k_{N,t} (x;y) = (Q_{N,t} \otimes Q_{N,t}) \left[- N \omega_N (x-y) \ph^2_{N,t} ((x+y)/2) \right] \end{equation}
where $Q_{N,t} = 1 - |\ph_{N,t} \rangle \langle \ph_{N,t}|$ is the orthogonal projection onto the orthogonal complement of the solution of the modified Hartree equation (\ref{eq:NLSN}). 

Let us briefly explain the choice (\ref{eq:Bog-trans}), (\ref{eq:kNt}). Since $T_{N,t}$ aims at generating correlations, it is natural to define its kernel $k_{N,t}$ through the solution of (\ref{eq:Neum}). In particular, the choice (\ref{eq:Bog-trans}) guarantees a crucial cancellation in the generator of the fluctuation dynamics, defined in (\ref{def:wcG}), which allows us to show the bounds (\ref{eq:G2Nt-est}) in Prop. \ref{lem:wG}. The cancellation is hidden in Prop. \ref{lem:cLj} and leads to the estimates (\ref{eq:bndcE}). It combines the quadratic term on the fourth line on the r.h.s. of (\ref{cLj}) with contributions arising from conjugation of the kinetic energy $d\Gamma (-\Delta)$ and of the quartic interaction on the last line of (\ref{cLj}) with $T_{N,t}$, reconstructing (\ref{eq:Neum}). 

It is important to observe that (\ref{eq:kNt}) is the integral kernel of a Hilbert-Schmidt operator. Abusing notation and denoting by $k_{N,t}$ both the Hilbert-Schmidt operator and its integral kernel, we easily find (using (\ref{eq:fN-bound}) and (\ref{eq:phi-bds})) 
\begin{equation}\label{eq:kNt-bd} \begin{split} \| k_{N,t} \|_\text{HS} &= \| k_{N,t} \|_2 \leq C  \\
\| \nabla k_{N,t} \|_\text{HS} &= \| k_{N,t} \nabla \|_{\text{HS}} = \| \nabla_1 k_{N,t} \|_2 = \| \nabla_2 k_{N,t} \|_2 \leq C N^{\beta/2}.
\end{split} \end{equation}
These bounds reflect the idea that, through $T_{N,t}$, we only produce a bounded number of excitations, causing however a large change in the energy. 

Notice that the action of the Bogoliubov transformation (\ref{eq:Bog-trans}) on creation and annihilation operators is explicit. For any $f \in L^2_{\perp \ph_{N,t}} (\bR^3)$, we find 
\[ \begin{split}  T_{N,t} a (f) T_{N,t}^* &= a (\cosh_{k_{N,t}} (f)) + a^* (\sinh_{k_{N,t}} (\bar{f})) 
\\  T_{N,t} a^* (f) T_{N,t}^* &= a^* (\cosh_{k_{N,t}} (f)) + a (\sinh_{k_{N,t}} (\bar{f})) \end{split} \]
where $\cosh_{k_{N,t}}$ and $\sinh_{k_{N,t}}$ are the linear operators defined by the absolutely convergent series 
\begin{equation}\label{eq:chsh} \begin{split} \cosh_{k_{N,t}} &= \sum_{n \geq 0} \frac{1}{(2n)!} \,(k_{N,t}  \overline{k}_{N,t})^n \, , \; \qquad    \sinh_{k_{N,t}} = \sum_{n \geq 0} \frac{1}{(2n+1)!} (k_{N,t} \overline{k}_{N,t})^n k_{N,t}\, . \end{split} \end{equation}

Using the Bogoliubov transformation $T_{N,t}$ to implement correlations, one can construct norm approximations for the many-body evolution that are valid also in regimes 
with $\beta > 1/2$. For Fock space initial data, it was recently proven in \cite{BCS} that, for every $0 < \beta < 1$ and for every $N$ large enough, there exists a unitary evolution $\cU_{2,N}^\beta$ with a time-dependent generator quadratic in creation and annihilation operators, such that  
\[ \left\| e^{-i \cH_N t} W(\sqrt{N} \ph)  T^*_{N,0} \Omega - W (\sqrt{N} \ph_{N,t}) T^*_{N,t} \, \cU_{2,N}^\text{f} (t;0) \Omega \right\| \to 0 \]
as $N \to \infty$ (to be more precise, in \cite{BCS}, the kernel $k_{N,t}$ was chosen slightly different, without the orthogonal projection $(Q_{N,t} \otimes Q_{N,t})$). In other words, for initial data of the form $W(\sqrt{N} \ph) T_{N,0} \Omega$, describing an approximate coherent state, modified by the Bogoliubov transformation $T_{N,0}$ to take into account correlations, the full many-body time-evolution can be approximated in terms of the modified $N$-dependent Hartree equation (\ref{eq:NLSN}) (describing the dynamics of the condensate), of the Bogoliubov transformation (\ref{eq:Bog-trans}) (generating the correlation structure) and of the quadratic evolution $\cU_{2,N}^\text{f}$ (which, similarly to (\ref{eq:U2-action}), also acts as a time-dependent Bogoliubov transformation). Using a related approach, a similar result has been established in \cite{GM2} 
for $\beta < 2/3$. 

Our aim in the present paper is to obtain a norm-approximation for the many-body evolution of $N$-particle initial data exhibiting BEC for the whole range of parameters $0<\beta<1$. To reach this goal, we will combine ideas from \cite{LNS} and \cite{NN1,NN2} with ideas from \cite{BCS}, in particular, with the idea of using Bogoliubov transformations of the form (\ref{eq:Bog-trans}) to implement correlations. To state our main result, we 
define the unitary dynamics $\cU_{2,N} (t;s)$ as the two-parameter unitary group on the Fock space $\cF$ satisfying 
\begin{equation}\label{eq:U2Nts-def} i\partial_t \cU_{2,N} (t;s) = \cG_{2,N,t} \,  \cU_{2,N} (t;s), \qquad \cU_{2,N} (s;s) = \1_{\cF} \end{equation}
with the time-dependent quadratic generator $\cG_{2,N,t}$ given by 
\begin{equation} \label{eq:cG2N}
\begin{split}
\cG_{2,N,t} & = \eta_N(t) + (i\partial_t T_{N,t})T_{N,t}^*  +\cG^{\cV}_{2,N,t}+\cG^{\cK}_{2,N,t}+\cG^{\lambda_N}_{2,N,t}
\end{split}
\end{equation}
with the phase $\eta_N(t)$ defined by 
    \begin{equation} \label{eq:etaN}
    \begin{split}
    \eta_N(t) & = \frac{N+1}{2} \product{\varphi_{N,t}}{ [V_N (1-2f_N) \ast|\varphi_{N,t}|^2 ]\varphi_{N,t}}-\mu_{N}(t) \\
    &\hspace{0.5cm}+ \int dx\; \big(V_N\ast |\varphi_{N,t}|^2\big)(x)\|\text{sh}_x\|^2  + \int dx \; \langle\nabla_x \text{sh}_x,\nabla_x \text{sh}_x\rangle \\
    &\hspace{0.5cm}+ \int dx dy\; K_{1,N,t}(x;y)\langle \text{sh}_x, \text{sh}_y\rangle  + \Re \int dx dy\; K_{2,N,t}(x;y)\; \langle \text{sh}_x, \text{ch}_y\rangle\\
    &\hspace{0.5cm} +\frac{1}{2N}\int dx dy\; V_N(x-y)\big|\big\langle \text{sh}_x - \varphi_{N,t}(x)\text{sh}(\varphi_{N,t}), \text{ch}_y - \varphi_{N,t}(y)\text{ch}(\varphi_{N,t})\big\rangle \big|^2\end{split} \end{equation}
with $\mu_{N}(t) = \product{\varphi_{N,t}}{ [(V_N \omega_{N}) \ast|\varphi_{N,t}|^2 ]\varphi_{N,t}}$  
and where the operators $\cG^{\cV}_{2,N,t}$, $\cG^{\lambda_N}_{2,N,t}$ and $\cG^{\cK}_{2,N,t}$ are given by 
		\begin{equation}\label{eq:defG2NtVlambda}\begin{split}
        \cG^{\cV}_{2,N,t}=&\; \int dx\; \big(V_N\ast |\varphi_{N,t}|^2\big)(x) \big[ a^*(\text{ch}_x)a(\text{ch}_x) + a^*(\text{ch}_x)a^*(\text{sh}_x) \\
        &\hspace{2cm} + a(\text{ch}_x)a(\text{sh}_x)+a^*(\text{sh}_x)a(\text{sh}_x)\big]  \\
        &\;+ \int dxdy\; K_{1,N,t}(x;y) \big[ a^*(\text{ch}_x)a(\text{ch}_y) + a^*(\text{ch}_x)a^*(\text{sh}_y) \\
        &\hspace{2cm} + a(\text{ch}_y)a(\text{sh}_x)+a^*(\text{sh}_y)a(\text{sh}_x)\big]  \\
        &\;+  \frac12\int dxdy\; K_{2,N,t}(x;y)\big[ a^*_x a^*(\text{p}_y) + a^*_x a(\text{sh}_y)
        +a^*(\text{p}_x) a^*(\text{p}_y) + a^*(\text{p}_x) a(\text{sh}_y)\\
        &\hspace{2cm}+a^*_y a^*(\text{p}_x)+a^*_y a(\text{sh}_x) +a^*(\text{p}_y) a(\text{sh}_x)+ a(\text{sh}_x)a(\text{sh}_y)  +\text{h.c.} \big] \\
        \;&+\frac12 \Big[ \langle \varphi_{N,t}, V_N\ast|\varphi_{N,t}|^2\varphi_{N,t}  \rangle a^*(\varphi_{N,t})a^*(\varphi_{N,t}) \\
        \;&\hspace{2cm}- 2a^*(\varphi_{N,t})a^*\big([V_N\ast|\varphi_{N,t}|^2]\varphi_{N,t}\big) +\text{h.c.}\Big], \\
        \cG^{\lambda_N}_{2,N,t}=&\; N\lambda_N \int dx dy\; f_N(x-y)\chi(|x-y|\leq \ell)\varphi_{N,t}^2((x+y)/2) a_x^*a_y^* +\text{h.c.} \\
        \end{split}\end{equation}
and
		\begin{equation}\label{eq:defG2NtK}\begin{split}
        \cG^{\cK}_{2,N,t}=&\; \int dx\; \big[a_x^*(-\Delta_x)a_x+a^*_x a(-\Delta_x\text{p}_x)+a^*_x a^*(-\Delta_x\text{v}_x)+ a^*_x a^*(-\Delta_x\text{r}_x)   \\
        &\; \hspace{1cm}+ a^*(-\Delta_x \text{p}_x) a(\text{ch}_x) + a^*(-\Delta_x \text{p}_x) a^*(\text{sh}_x)+ a(-\Delta_x \text{r}_x) a_x \\
        &\; \hspace{1cm}+a(-\Delta_x\text{v}_x)a_x +  a( \text{sh}_x) a(-\Delta_x\text{p}_x)  + a^*(-\Delta_x \text{r}_x)a(\text{k}_x) \\
        &\;\hspace{1cm}+ a^*(-\Delta_x \text{r}_x)a(\text{r}_x)+ a^*(\text{k}_x) a(-\Delta_x \text{r}_x)+ a^*\big(\nabla_x\text{k}_x\big) a\big(\nabla_x \text{k}_x\big) \big] \\
        & +\frac12 \int dxdy\; N\omega_N(x-y)  \big[\varphi_{N,t}((x+y)/2) \Delta\varphi_{N,t}((x+y)/2) \\ &\; \hspace{3cm} +\nabla\varphi_{N,t}((x+y)/2)\cdot\nabla\varphi_{N,t}((x+y)/2) \big] a_x^*a_y^*  + \text{h.c.} 
        \end{split}\end{equation}
Here we have introduced the notation
\begin{align*}
K_{1,N,t}&=Q_{N,t}\w K_{1,N,t}Q_{N,t}\\
K_{2,N,t}&=Q_{N,t}\otimes Q_{N,t} \w K_{2,N,t}
\end{align*}
where $\w K_{1,N,t}$ is the operator on $L^2(\R^3)$ with integral kernel 
$$\w K_{1,N,t}(x,y)=\varphi_{N,t}(x) V_N(x-y)\overline{\varphi_{N,t}(y)}$$
and $\w K_{2,N,t}$ is a function in $L^2(\R^3\times \R^3)$:
$$\w K_{2,N,t}(x,y)=\varphi_{N,t}(x)V_N(x-y)\varphi_{N,t}(y).$$

Finally, we also use the notation $j_x(\cdot):= j(\cdot;x)$ for any $j\in L^2(\mathbb{R}^3\times \mathbb{R}^3)$. Moreover, with (\ref{eq:chsh}), we set \[ \text{sh} = \sinh_{k_{N,t}}, \quad \quad \text{ch} = \cosh_{k_{N,t}} \]  and we decompose $\text{sh}=\text{k} + \text{r}$ and $\text{ch}=\1 + \text{p}$ as well as
		\[\begin{split}
        k_{N,t}(x;y)&= -N\omega_N(x-y)\varphi_{N,t}^2((x+y)/2)+\text{v}(x;y);\hspace{0.5cm}\forall x,y\in\mathbb{R}^3.
         \end{split}\]

We are now ready to state our first main result, providing a norm-approximation for the many-body evolution of $N$-particle initial data exhibiting BEC. To this end, let us first collect some conditions that will be required throughout the paper.
 
\medskip

{\it Hypothesis A:} We assume that $0 < \beta  <1$. We suppose, moreover, the interaction potential $V$ to be smooth, radially symmetric, compactly supported and pointwise non-negative. Furthermore, we choose $f_N$ to be the solution of the Neumann problem (\ref{eq:Neum}) on the ball $|x| \leq \ell$, for a sufficiently small\footnote{The smallness of $\ell$ is used because it implies that the kernel $k_{N,t}$ introduced in (\ref{eq:kNt}) has a small Hilbert-Schmidt norm; this in turn implies that conjugation with the Bogoliubov transformation  $T_{N,t}$ produces only small changes in the number of particles operator; see Proposition \ref{lem:Bog-N}.} (but fixed, independent of $N$) parameter $\ell > 0$. Finally, we let $\ph_{N,t}$ be the solution of the $N$-dependent nonlinear Hartree equation (\ref{eq:NLSN}) with initial data $\ph_0 \in H^4 (\bR^3)$. 

\medskip

{\it Remark:} We need $V$ to have compact support to study the solution of (\ref{eq:Neum}) and to establish the bounds (\ref{eq:fN-bound}), following \cite{BCS}. For the same reason (but also for the many-body analysis), we need some smoothness of $V$. The assumption $V \in L^3 (\bR^3)$ is sufficient for our purposes; we do not aim at optimal conditions, here. The assumption $\ph_0 \in H^4 (\bR^3)$ allows us to show the bounds (\ref{eq:phi-bds}) for the solution of (\ref{eq:NLSN}); these estimates play an important role in the analysis of the many-body dynamics (in particular, in the proof of Prop. \ref{lem:wG}, where we need control of $\| \partial_t \ph_{N,t} \|_\infty$). One may be able to partially relax this assumption by using space-time norms; also here, we do not aim at optimal conditions. 

\begin{theorem}\label{thm:main1}
Assume that Hypothesis A holds true. Let $\xi_N \in \cF_{\perp \ph_0}$ with $\| \xi_N \| = 1$ and 
\begin{equation}\label{eq:ass-xi} \langle \xi_{N} , (\cK+ \cN) \xi_{N} \rangle \leq C. \end{equation}
Let $\Psi_{N,t}$ be the solution of the Schr\"odinger equation (\ref{eq:schr}) with initial data
\begin{equation}\label{eq:psiN0-thms}
\Psi_{N,0} = U_{\varphi_{0}}^* \1^{\le N} T_{N,0}^* \xi_{N} 
\end{equation}
and let $\cU_{2,N} (t;s)$ be the unitary dynamics on $\cF$ defined in  (\ref{eq:U2Nts-def}). Then, for all $\alpha < \min (\beta/2, (1-\beta)/2)$, there exists a constant $C > 0$ such that
\begin{align} \label{eq:mainN-intro}
\left\| U_{\ph_{N,t}} \Psi_{N,t} - T^*_{N,t} \, \cU_{2,N} (t;0) \, \xi_{N} \right\|^2 \le C N^{-\alpha} \exp(C\exp(C|t|)) 
\end{align}
for all $N$ sufficiently large and all $t \in \bR$. 
\end{theorem}

Since the quadratic evolution $\cU_{2,N} (t;s)$ depends on $N$, it is natural to ask what 
happens as $N \to \infty$. Proceeding similarly to \cite{BCS}, we observe that $k_{N,t}$ can be approximated, for large $N$, by the limiting kernel  
    \begin{equation} \label{eq:kinfty}
    k_{t}(x;y)= (Q_{t}\otimes Q_{t}) \left[ - \omega_\infty(x-y) \varphi_{t}^2((x+y)/2)  \right] 
    \end{equation}
where $\varphi_t$ is the solution of the nonlinear Schr\"odinger equation (\ref{eq:NLS0}), $ Q_t = \1-|\varphi_t\rangle\langle \varphi_t|$ is the projection onto the orthogonal complement of $\ph_t$ and where $\omega_\infty$ is given by
		\begin{equation}\label{eq:winfty} \omega_{\infty}(x) := \left \{  \begin{array}{ll} \frac{b_0}{8\pi}\bigg[\frac1{|x|}-\frac{3}{2\ell}+\frac{|x|^2}{2\ell^3}\bigg] &\hspace{0.5cm}  \text{for }|x|\leq \ell,  \\  0 &\hspace{0.5cm} \text{for } |x|> \ell      \end{array}  \right.\end{equation}
where $b_0=\int V(x)\,dx$. With $ k_{t}$, we can define a new Bogoliubov transformation 
    \begin{equation} \label{eq:Bog-transinfty}
    T_{t} = \exp \left[ \frac 1 2\int dx dy \,k_{t} (x;y) a_x a_y - \text{h.c.} \right] 
    \end{equation}
Replacing $\cosh_{k_{N,t}}$, $\sinh_{k_{N,t}}$, $\text{p}_{k_{N,t}}$ and $\text{r}_{k_{N,t}}$ by their counterparts $\cosh_{k_{t}}$, $\sinh_{k_{t}}$, $\text{p}_{k_{t}}$ and $\text{r}_{k_{t}}$, replacing $\varphi_{N,t}$ by $\varphi_t$, the convolution $ V_N\ast(\cdot)$ by $b_0 \delta\ast(\cdot)$, the eigenvalue $N\lambda_N$ by its first order approximation $3b_0/(8\pi\ell^3)$, $ N\omega_N $ by $\omega_\infty$ and, finally, replacing $f_N=1-\omega_N $ by $f_\infty = 1$ in the 
operators $\cG_{2,N,t}^\cV, \cG_{2,N,t}^{\lambda_N}, \cG_{2,N,t}^\cK$ in (\ref{eq:defG2NtVlambda}) and (\ref{eq:defG2NtK}), we can define limiting operators $\cG_{2,t}^\cV, \cG_{2,t}^\lambda, 
\cG_{2,t}^{\cK}$ and we can use them to define the limiting generator 
\begin{equation}\label{eq:G2t} \cG_{2,t} = (i\partial_t T_t) T_t^* + \cG_{2,t}^\cV + \cG_{2,t}^{\cK} + \cG_{2,t}^\lambda \end{equation}
and the corresponding limiting fluctuation dynamics $\cU_{2}$ by  
\begin{equation}\label{eq:U2infty} i\partial_t \cU_{2} (t;s) = \cG_{2,t} \, \cU_{2} (t;s) \hspace{0.5cm} \cU_{2} (s;s) = \1_{\cF} \end{equation}
We are now ready to state our second main result.
\begin{theorem}\label{thm:main2}
Assume that Hypothesis A holds true. Let $\xi_N \in \cF_{\perp \ph_0}$ with $\| \xi_N \| = 1$ and (\ref{eq:ass-xi}). Let $\Psi_{N,t}$ be the solution of the Schr\"odinger equation (\ref{eq:schr}) with initial data (\ref{eq:psiN0-thms}) and let $\cU_{2} (t;0)$ be the unitary dynamics on $\cF$ defined in  (\ref{eq:U2infty}). Then, for all $\alpha < \min (\beta/2, (1-\beta)/2)$, there exists a constant $C > 0$ such that
    \begin{equation} \label{eq:maininfty-intro}
    \begin{split}
    \big\| U_{\ph_{N,t}} \Psi_{N,t} - e^{-i\int_0^t d\tau\; \eta_N (\tau)} \, T^*_{N,t} \, \cU_{2} (t;0) \, \xi_N \big\|^2 \leq  C N^{-\alpha} \exp(C\exp(C|t|)) 
    \end{split}
    \end{equation}
for all $N$ sufficiently large and all $t \in \bR$. 
\end{theorem}
        
Theorem \ref{thm:main1} and Theorem \ref{thm:main2} apply to the study of the time-evolution of initial data of the form \begin{equation}\label{eq:psiN0} \psi_{N,0} =  U_{\varphi_{0}}^* \1^{\le N} T_{N,0}^* \xi_{N}\end{equation} 
for a $\xi_N \in \cF_{\perp \ph_0}$ satisfying the bound
\begin{equation}\label{eq:xiN-bd2} \langle \xi_N , \left[ \cK + \cN \right] \xi_N \rangle \leq C \end{equation}
uniformly in $N$. It is natural to ask under which assumptions on $\psi_{N,0}$ is it possible to find $\xi_N \in \cF_{\perp \ph_0}$ such that (\ref{eq:psiN0}) and (\ref{eq:xiN-bd2}) hold true. The answer is given in our last main theorem.
\begin{theorem}\label{thm:main3}
Assume Hypothesis A holds true. Let $\Psi_{N,0} \in L^2_s (\bR^{3N})$ with reduced one-particle density matrix $\gamma_{N,0}$ such that 
\begin{equation}\label{eq:ass-gammaN} \tr \, \left| \gamma_{N,0} - |\ph_0 \rangle \langle \ph_0| \right| \leq C N^{-1}  
\end{equation}
and  
\begin{equation}\label{eq:ass-ener} \left| \frac{1}{N} \langle \Psi_{N,0} , H_N \Psi_{N,0} \rangle -  \left[ \| \nabla \ph_0 \|^2 + \frac{1}{2} \langle \ph_0, (V_N f_n * |\ph_0|^2) \ph_0 \rangle \right]  \right| \leq C N^{-1} \end{equation}
Let $\Psi_{N,t}$ be the solution of the Schr\"odinger equation (\ref{eq:schr}) with initial data $\psi_{N,0}$ and let $\cU_{2} (t;0)$ be the unitary dynamics on $\cF$ defined in (\ref{eq:U2infty}). Then, for all $\alpha < \min (\beta/2, (1-\beta)/2)$, there exists a constant $C > 0$ such that
\begin{equation} \label{eq:maininfty2}
\begin{split}
\big\| T_{N,t} U_{\ph_{N,t}} \Psi_{N,t} - e^{-i\int_0^t d\tau\; \eta_N (\tau)} \, \cU_{2} (t;0) \, T_{N,0} \, &U_{\ph_{N,0}} \Psi_{N,0}  \big\|^2 \\ &\leq  C N^{-\alpha} \exp(C\exp(C|t|)) 
\end{split}
\end{equation}
for all $N$ sufficiently large and all $t \in \bR$. 
\end{theorem}
        
{\it Remarks:} 
\begin{itemize}
\item[1)] Recall that, although this is not reflected in our notation, the family of Bogoliubov transformations $T_{N,t}$ and the quadratic evolutions $\cU_{2,N} (t;0)$ in Theorem~\ref{thm:main1} and $\cU_2 (t;0)$ in Theorem \ref{thm:main2} and in Theorem \ref{thm:main3} depend on the choice of the length scale $\ell > 0$ in (\ref{eq:Neum}). This parameter is chosen small enough, but fixed. 
\item[2)] The bounds (\ref{eq:mainN-intro}), (\ref{eq:maininfty-intro}) and (\ref{eq:maininfty2}) give norm approximations of the full many-body dynamics of initial data exhibiting BEC in terms of Fock space dynamics $\cU_{2,N} (t;0)$ or $\cU_2 (t;0)$ with quadratic generators, of the family of time-dependent Bogoliubov transformation $T_{N,t}$ and of the solution $\ph_{N,t}$ of the modified Hartree equation. 
\item[3)] We assumed the bounds (\ref{eq:ass-gammaN}) and (\ref{eq:ass-ener}) to hold with best possible rates $N^{-1}$, corresponding to initial data with bounded (i.e. $N$-independent) number of excitations and with bounded excitation energy. One could relax a bit this requirement allowing for more excitations and for a larger excitation energy but then, of course, the rate on the r.h.s. of (\ref{eq:maininfty2}) would deteriorate. 
\item[4)] From the analysis of \cite[Section 6]{BS}, it is clear that one can also replace the condition (\ref{eq:ass-gammaN}) by the weaker bound
\begin{equation}\label{eq:ass-cond2} 1 - \langle \ph_0 , \gamma_{N,0} \ph_0 \rangle \leq C N^{-1} \end{equation}
if one additionally assumes that there exists a sufficiently regular external confining potential $V_\text{ext}$ such that $\ph_0$ minimizes the energy functional 
\begin{equation}\label{eq:ass-ener2}\begin{split} \cE (\ph) = \; &\int \left[ | \nabla \ph (x)|^2 + V_\text{ext} (x) |\ph (x)|^2 \right] dx \\ &+ \frac{1}{2} \int dx dy V_N (x-y) f_N (x-y) |\ph (x)|^2 |\ph (y)|^2  
\end{split}\end{equation}
with the constraint $\| \ph \| =1$ and if one replaces the condition (\ref{eq:ass-ener}) by the similar bound 
\[ \left| \frac{1}{N} \langle \psi_{N,0} , H^\text{trap}_N \psi_{N,0} \rangle - \cE (\ph_0) \right| \leq C N^{-1} \]
for the Hamilton operator with confining potential $H_N^\text{trap} = H_N + \sum_{j=1}^N V_\text{ext} (x_j)$. The assumptions (\ref{eq:ass-cond2}), (\ref{eq:ass-ener2}) are expected to hold true if $\psi_{N,0}$ is the ground state of the trapped Hamiltonian $H_N^\text{trap}$. They describe experiments where particles are initially trapped by external fields and they are cooled down at temperatures so low that they essentially relax to the ground state. 
\item[5)] The conditions (\ref{eq:ass-cond2})-(\ref{eq:ass-ener2}), and hence  (\ref{eq:psiN0})-(\ref{eq:xiN-bd2}), have been proved rigorously for the ground states (more generally, low-lying eigenstates) of trapped systems  when either $\beta=0$ (mean-field regime) \cite{Seiringer-11,GreSei-13,LNSS,DerNap-14,NamSei-15,Pizzo-15}, or $0<\beta<1$ and particles are trapped in a unit torus without an external potential \cite{BBCS0,BBCS}. 
\end{itemize}

{\it Acknowledgements.} We gratefully acknowledge support from the  Swiss National Foundation of Science through the NCCR SwissMAP and the SNF Grant ``Dynamical and energetic properties of Bose-Einstein condensates'' (B.S.)  and from the Polish National Science Center (NCN) grant No. 2016/21/D/ST1/02430 (M.N.).

\section{Outline of the proof}

In this section we explain the overall strategy of the proof. As in Theorem \ref{thm:main1}, we denote by $\Psi_{N,t}$ the solution of the $N$-particle Schr\"odinger equation (\ref{eq:schr}) with the initial data $\Psi_{N,0}= U_{\varphi_{N,0}}^* \1^{\le N} T_{N,0}^* \xi_{N}$, where 
$\xi_N \in \cF^{\leq N}_{\perp \ph}$ is such that 
\[ \langle \xi_N, (\cN + \cK) \xi_N \rangle \leq C \]
uniformly in $N$. Furthermore, we denote by $\ph_{N,t}$ the solution of the modified, $N$-dependent, nonlinear Hartree equation (\ref{eq:NLSN}), with initial data $\ph_0 \in H^4 (\bR^3)$. 

\subsection{Fluctuation evolution} First of all, we apply the map $U_{\ph_{N,t}}$, defined in (\ref{eq:U-def}), to $\Psi_{N,t}$. This allows us to remove the condensate described at time $t$ by $\ph_{N,t}$ and to focus on the orthogonal fluctuations. We set 
\begin{equation}\label{eq:phiNt}
\Phi_{N,t} = U_{\varphi_{N,t}}\Psi_{N,t},
\end{equation}
and we observe that $\Phi_{N,t} \in \cF^{\leq N}_{\perp \ph_{N,t}}$ satisfies the equation  
\begin{equation} \label{eq:Phi}
i\partial_t \Phi_{N,t} = \cL_{N,t} \Phi_{N,t}, 
\end{equation}
with the initial data $\Phi_{N,0}= \1^{\le N} T_{N,0}^* \xi_{N}$ and the generator 
\begin{equation}\label{eq:cLNt} \cL_{N,t} = (i\partial_t U_{\varphi_{N,t}} ) U^*_{N,t} + U_{\varphi_{N,t}} H_N U_{\varphi_{N,t}}^*. \end{equation}

Using (\ref{eq:U-rules}) and computing the first term on the r.h.s. of (\ref{eq:cLNt}) as in   \cite{LNS}, we obtain 
\begin{equation}\label{eq:cLNt-2}
\begin{split}  
&\cL_{N,t} =  \frac{N+1}{2} \product{\varphi_{N,t}}{ [V_N (1-2f_N) \ast|\varphi_{N,t}|^2 ]\varphi_{N,t}} - \mu_{N}(t) \\
&+ \frac12\product{\varphi_{N,t}}{ [V_N\ast|\varphi_{N,t}|^2]\varphi_{N,t} }  \frac{\cN(\cN+1)}{N}  \\
 &+\Big[ \sqrt{N}\Big[ a^* (Q_{N,t}[(V_N \omega_N) \ast|\varphi_{N,t}|^2] \varphi_{N,t}) - a^* (Q_{N,t}[V_N\ast|\varphi_{N,t}|^2] \varphi_{N,t}) \frac{\cN}{N}  \Big] \sqrt{\frac{N-\cN}{N}}\\ &\hspace{3cm} +{\rm h.c.} \Big] \\
 &+ \dGamma\Big(  -\Delta + (V_Nf_N)*|\varphi_{N,t}|^2+ K_{1,N,t} -\mu_{N,t} \Big)  \\
 &+ \dGamma\Big( Q_{N,t}(V_N \omega_N *|\varphi_{N,t}|^2)Q_{N,t} \Big) - \dGamma\Big( Q_{N,t}(V_N *|\varphi_{N,t}|^2)Q_{N,t} + K_{1,N,t}\Big)\frac{\cN}{N}\\
   & + \bigg[\frac{1}{2} \int  \d x \d y \, K_{2,N,t}(x,y)a^*_x a^*_y  \frac{\sqrt{(N-\cN)(N-\cN-1)}}{N} + {\rm h.c.} \bigg] \\
 & + \bigg[\frac{1}{\sqrt{N}}\int dx dy dx'dy'\ \, (Q_{N,t} \otimes Q_{N,t} V_N Q_{N,t} \otimes 1)(x,y;x',y')\varphi_{N,t} (y') a_x^* a_y^* a_{x'} \sqrt{\frac{N-\cN}{N}}\\ &\hspace{3cm} + \text{h.c.} \bigg]\\
 &+  \frac{1}{2N}\int dx dy dx'dy'\ \, (Q_{N,t}\otimes Q_{N,t}V_N Q_{N,t}\otimes  Q_{N,t}) (x,y;x',y')a^*_x a^*_y a_{x'} a_{y'}
\end{split}
\end{equation} 
with
\begin{align*} 
\mu_{N}(t) &:=  \product{\varphi_{N,t}}{ [(V_N \omega_{N}) \ast|\varphi_{N,t}|^2 ]\varphi_{N,t}} .
\end{align*}

\subsection{Modified fluctuation evolution}

Next, we have to remove the singular correlation structure from $\Phi_{N,t}$. Since $\Psi_{N,t} = U^*_{\ph_{N,t}} \Phi_{N,t}$ and since $U^*_{\ph_{N,t}}$ just adds products of solutions of the nonlinear equation (\ref{eq:NLSN}), it is clear that all 
correlations developed by $\Psi_{N,t}$ must be contained in $\Phi_{N,t}$. As a consequence, 
at least for $\beta > 1/2$, the time evolution of $\Phi_{N,t}$ cannot be generated by a 
quadratic Hamiltonian, not even approximately in the limit of large $N$. To remove 
correlations from $\Phi_{N,t}$ we would like to follow the idea of \cite{BCS} and apply the Bogoliubov transformation $T_{N,t}$ defined in (\ref{eq:Bog-trans}). Unfortunately, $T_{N,t}$ does not preserve the number of particles, and therefore it does not leave the truncated Fock space $\cF_{\perp \ph_{N,t}}^{\leq N}$ invariant. Since $T_{N,t}$ only creates few particles (the bound (\ref{eq:kNt-bd}) implies that $T_{N,t} \cN T_{N,t}^* \leq C \cN$), this should not be a serious obstacle. To circumvent it, it seems natural to give up the restriction on the number of particles and consider $\Phi_{N,t}$ as a vector in the untruncated Fock space $\cF_{\perp \ph_{N,t}}$. The drawback of this approach is the fact that the generator $\cL_{N,t}$ computed in (\ref{eq:cLNt-2}) is defined only on sectors with at most $N$ particles. So, we proceed as follows; first we approximate $\Phi_{N,t}$ by a new, modified, fluctuation vector $\wt{\Phi}_{N,t}$, whose dynamics is governed by a modified generator $\wt{\cL}_{N,t}$ which, on the one hand, is close to $\cL_{N,t}$ when acting on vectors with a small number of particles and, on the other hand, is well-defined on the full untruncated Fock space $\cF_{\perp \ph_{N,t}}$. To define $\wt{\cL}_{N,t}$ we proceed as follows. Starting from the expression on the r.h.s. of (\ref{eq:cLNt-2}), we replace first of all the factor $\sqrt{(N-\cN)(N-\cN-1)}$ by $N-\cN$; the error is small, since
\[ |\sqrt{(N-x)(N-x-1)} - (N-x)| \leq 1 \]
for all $x \in \bN$. 

Secondly, we replace $\sqrt{N-\cN}$ by $\sqrt{N} G_b(\cN/N)$ where
\bq \label{eq:Pb}
G_b(t):= \sum_{n=0}^{b} \frac{(2n)!}{(n!)^2 4^n (1-2n)} t^n.
\eq
Indeed, the polynomial $G_b(t)$ is the Taylor series for $\sqrt{1-t}$ around $t=0$; it 
satisfies
\bq \label{eq:Pb-1}
|\sqrt{1-t} - G_b(t)| \le C t^{b+1}, \quad \forall t\in [0,1].
\eq
for a constant $C>0$ depending on $b$. Here $b\in \mathbb{N}$ is a large, fixed number, that will be specified later. 

Finally, we add a term of the form $C_b e^{C_b |t|}  \cN (\cN/N)^{2b}$ with a sufficiently large constant $C_b$ that will also be specified later. Since the generators $\cL_N$ and $\wt{\cL}_N$ will act on states with small number of particles, we expect this term to have a negligible effect on the dynamics (on the other hand, it allows us to better control the energy). With these changes, we obtain the modified generator
\begin{equation} \label{eq:wLNt}
\begin{split}  
&\w\cL_{N,t} = \frac{N+1}{2} \product{\varphi_{N,t}}{ [V_N (1-2f_N) \ast|\varphi_{N,t}|^2 ]\varphi_{N,t}} - \mu_{N}(t) \\
&+ \frac12\product{\varphi_{N,t}}{ [V_N\ast|\varphi_{N,t}|^2]\varphi_{N,t} }  \frac{\cN(\cN+1)}{N}  \\
 &+\Big[ \sqrt{N}\Big[ a^* (Q_{N,t}[(V_N \omega_N) \ast|\varphi_{N,t}|^2] \varphi_{N,t}) - a^* (Q_{N,t}[V_N\ast|\varphi_{N,t}|^2] \varphi_{N,t}) \frac{\cN}{N}  \Big] G_b(\cN/N) \\ &\hspace{3cm} +{\rm h.c.} \Big] \\
 &+ \dGamma\Big(  -\Delta + (V_Nf_N)*|\varphi_{N,t}|^2 + K_{1,N,t} -\mu_{N,t} \Big)  \\
 &+ \dGamma\Big( Q_{N,t}(V_N \omega_N *|\varphi_{N,t}|^2)Q_{N,t} \Big) - \dGamma\Big( Q_{N,t}(V_N *|\varphi_{N,t}|^2)Q_{N,t} + K_{1,N,t}\Big)\frac{\cN}{N}\\
   & + \bigg[\frac{1}{2} \int  \d x \d y \, K_{2,N,t}(x,y)a^*_x a^*_y  \frac{N-\cN}{N} + {\rm h.c.} \bigg] \\
 & + \bigg[\frac{1}{\sqrt{N}}\int dx dy dx'dy'\ \, (Q_{N,t} \otimes Q_{N,t} V_N Q_{N,t} \otimes 1)(x,y;x',y')\varphi_{N,t} (y') a_x^* a_y^* a_{x'} G_b(\cN/N) \\ &\hspace{3cm} + \text{h.c.}\bigg] \\
 &+  \frac{1}{2N}\int dx dy dx'dy'\ \, (Q_{N,t}\otimes Q_{N,t}V_N Q_{N,t}\otimes  Q_{N,t})(x,y;x',y')a^*_x a^*_y a_{x'} a_{y'} \\ &+ C_b e^{C_b |t|} \, \cN (\cN/N)^{2b}.
\end{split} 
\end{equation} 
Using this modified generator, we define the modified fluctuation dynamics $\wt{\Phi}_{N,t}$ as the solution of the Schr\"odinger equation 
\bq \label{eq:wPhi}
i\partial_t \w\Phi_{N,t} = \w\cL_{N,t} \w\Phi_{N,t}, 
\eq 
with the initial data $\w\Phi_{N,0} = T_{N,0}^* \xi_{N}$. Observe that 
$\w\Phi_{N,t}\in \cF_{\bot \varphi_{N,t}}$. Indeed, arguing as in \cite[Lemma 9]{LNS}, we have
		\begin{equation} \label{eq:dt-wPhiN}
        \begin{split} \frac{d}{dt} \|a(\varphi_{N,t})\w\Phi_{N,t}\|^2 =& i\big\langle\w\Phi_{N,t},\big[\w\cL_{N,t}, a^*(\varphi_{N,t})a(\varphi_{N,t})\big] \w\Phi_{N,t}\big\rangle \\
        &+2\Im\big\langle\w\Phi_{N,t},a^*(i\partial_t \varphi_{N,t})a(\varphi_{N,t}) \w\Phi_{N,t}\big\rangle =0,
        \end{split}
        \end{equation}
because, using that $[a^*(\varphi_{N,t})a(\varphi_{N,t}),\cN]=0 $, we find
		\[\begin{split} \big[\w\cL_{N,t}, a^*(\varphi_{N,t})a(\varphi_{N,t})\big]&= \big[\dGamma\big(-\Delta+ (V_N f_N)*|\varphi_{N,t}|^2  \big), a^*(\varphi_{N,t})a(\varphi_{N,t})\big]\\
        & = a^*\big( [-\Delta + (V_N f_N)*|\varphi_{N,t}|^2] \varphi_{N,t}\big)a(\varphi_{N,t})-\text{h.c.}\\
        & = a^*(i\partial_t \varphi_{N,t})a(\varphi_{N,t})-\text{h.c.}
        \end{split}\]
Notice moreover that we find it more convenient to choose the initial data for the modified dynamics slightly different from the initial data for the original fluctuation dynamics (we do not include the cutoff to $\cN \leq N$ in the definition of $\w\Phi_{N,0}$). Nevertheless, it is possible to prove that the two dynamics remain close; this is the content of the next lemma, which is the first step in the proof of Theorem \ref{thm:main1}. 
\begin{lemma} \label{lem:Phi-wPhi}
Assume Hypothesis A holds true. Let $\Phi_{N,t}$ be as defined in \eqref{eq:Phi} and $\w\Phi_{N,t}$ as in \eqref{eq:wPhi}. Here, we assume that the parameters $b \in \bN$ and $C_b > 0$ in (\ref{eq:wLNt}) are large enough, and that $\xi_N \in \cF_{\perp \ph_0}$ is such that $\| \xi_N \| \leq 1$ and 
\begin{equation}\label{eq:ass-xiNb} \langle \xi_N, \big[ \cH_N + \cN + \cN (\cN/N)^{2b} \big] \xi_N \rangle \leq C \end{equation}
uniformly in $N$. Then, for all $\alpha < (1-\beta)/2$, 
there exists a constant $C > 0$ such that 
\[ 
 \| \Phi_{N,t}-\widetilde\Phi_{N,t} \|^2 \le C N^{-\alpha} \exp(C\exp(C|t|)) 
\]
for all $t \in \bR$. 
\end{lemma}  

\subsection{Bogoliubov transformation} Next, we apply the Bogoliubov transformation (\ref{eq:Bog-trans}) to the modified fluctuation evolution $\w\Phi_{N,t}$ defined in \eqref{eq:wPhi}. We set 
\begin{equation} \label{eq:defxiNt}
\xi_{N,t}= T_{N,t} \w\Phi_{N,t}
\end{equation}
Then $\xi_{N,t} \in \cF_{\perp \varphi_{N,t}}$ (with no restriction on the number of particles) and it solves the Schr\"odinger equation
\bq \label{eq:eq-xiN}
i\partial_t \xi_{N,t} = \cG_{N,t} \xi_{N,t},  
\eq
with the initial data $\xi_{N,0}= \xi_N$ and the generator 
\bq 
\cG_{N,t}= (i\partial_t T_{N,t})T_{N,t}^* + T_{N,t} \w\cL_{N,t} T_{N,t}^*. \label{def:wcG}
\eq
As explained above, the application of the Bogoliubov transformation $T_{N,t}$ takes care of correlations and makes it possible for us to approximate the evolution (\ref{eq:eq-xiN}) with the unitary 
evolution $\cU_{2,N}$, having the quadratic generator (\ref{eq:cG2N}). This is the content of the next lemma. 
\begin{lemma} \label{lem:Compare-xiN} 
Assume Hypothesis A holds true. Let $\xi_{N,t}$ be defined as in (\ref{eq:defxiNt}) and $\xi_{2,N,t} = \cU_{2,N} (t;0) \xi_N$ with the unitary evolution $\cU_{2,N}$ defined in (\ref{eq:U2Nts-def}). Here, we assume that the parameters $b \in \bN$ and $C_b > 0$ in (\ref{eq:wLNt}) are large enough, and that $\xi_N \in \cF_{\perp \ph_0}$ is such that $\| \xi_N \| \leq 1$ and (\ref{eq:ass-xiNb}) holds true. Then there exists $C > 0$ such that  
$$
\| \xi_{N,t}- \xi_{2,N,t}\|^2 \le C N^{-\alpha}\exp(C\exp(C|t|)) , 
$$
for all $t \in \bR$, with $\alpha = \min (\beta /2 , (1-\beta)/2)$. 
\end{lemma}  

Theorem \ref{thm:main1} is a consequence of Lemma \ref{lem:Phi-wPhi} and Lemma \ref{lem:Compare-xiN}, up to the remark that the assumption on the sequence $\xi_N \in \cF_{\perp \ph_0}$ appearing in Theorem \ref{thm:main1} is weaker than the assumption (\ref{eq:ass-xiNb}) appearing in both lemmas. So, to conclude the proof of Theorem \ref{thm:main1}, we need an additional localization argument, which will be explained in Section \ref{sec:thm-proof}. 

To prove Theorem \ref{thm:main2} we will then compare $\xi_{2,N,t}$ with 
$\xi_{2,t} = \cU_{2} (t;0) \xi_N$, where $\cU_2$ is the limiting evolution defined in (\ref{eq:U2infty}), by controlling the difference between the two generators. 

Finally, Theorem \ref{thm:main3} will follow from Theorem \ref{thm:main2}, by proving that, under the assumptions (\ref{eq:ass-gammaN}) and (\ref{eq:ass-ener}), it is possible to write $\psi_{N,0} = U^*_{\ph_0} {\bf 1}^{\leq N} \, T_{N,0}^* \xi_N$ with a sequence $\xi_N \in \cF_{\perp \ph_0}$ satisfying the condition (\ref{eq:ass-xi}).

The rest of the paper is organized as follows. In Section \ref{sec:xi-xi2} we show Lemma \ref{lem:Compare-xiN}. In Section \ref{sec:phiwphi}, we prove Lemma \ref{lem:Phi-wPhi} making use of some energy estimates. Finally, in Section \ref{sec:thm-proof}, we conclude the proof of our three main theorems.

\section{Analysis of Bogoliubov transformed dynamics}
\label{sec:xi-xi2}

In this section, we prove Lemma \ref{lem:Compare-xiN}. To this end, we need to study the properties of the generator $\cG_{N,t}$ defined in (\ref{def:wcG}).
\begin{proposition} \label{lem:wG} 
Assume that Hypothesis A holds true. Then, there exists a constant $C > 0$ and, for every fixed $b \in \bN$, a constant $K_b > 0$ such that the generator $\cG_{N,t}$ in  (\ref{def:wcG}) can be written as   
\bq
\cG_{N,t}= \cG_{2,N,t} + \cV_N + C_b e^{C_b |t|}  \cN(\cN/N)^{2b} + 
\cE_{N,t}
\eq
with the quadratic generator $\cG_{2,N,t}$, defined as in \eqref{eq:cG2N}, satisfying the estimates
\begin{equation}\label{eq:G2Nt-est}
\begin{split} 
\pm (\cG_{2,N,t} - \eta_N (t) - \cK) &\leq C e^{C|t|} (\cN+1) \\
\pm \big[ \cG_{2,N,t} , i \cN \big] &\leq C e^{C|t|} (\cN+1) \\
\pm \partial_t \big( \cG_{2,N,t} - \eta_N (t) \big)  &\leq  C e^{C|t|} (\cN+1) 
\end{split} 
\end{equation}
and the error operator $\cE_{N,t}$ such that, with $\alpha = \min (\beta/2, (1-\beta)/2)$,  
\begin{equation}  \label{eq:wcEbounds}
\begin{split}
\pm \cE_{N,t} \le \; &\delta \cV_N + N^{-\beta/2} \cK + Ce^{C|t|}\max(N^{-\alpha}, \delta^{-1})(\cN+1) \\ &+ K_b e^{Ct}\max(\delta,\delta^{-1})\frac{(\cN+1)^2}{N} \\ &+\Big[ K_b \delta^{-1}e^{C|t|}+\frac{1}{2} C_b e^{C_b |t|} \Big]  (\cN+1)(\cN/N)^{2b}, \\ 
\pm i[\cN,\cE_{N,t}] \le \; &\delta \cV_N + N^{-\beta/2} \cK + Ce^{C|t|}\max(N^{-\alpha}, \delta^{-1})(\cN+1) \\ &+ K_b e^{C|t|}\max(\delta,\delta^{-1})\frac{(\cN+1)^2}{N} + K_b e^{C|t|}(\cN+1) (\cN/N)^{2b}, \\ 
\pm \partial_t \cE_{N,t} \le \; &\delta \cV_N + N^{-\beta/2} \cK + Ce^{Ct}\max(N^{-\alpha}, \delta^{-1})(\cN+1) \\ &+ K_b e^{C|t|}\max(\delta,\delta^{-1})\frac{(\cN+1)^2}{N} + K_b e^{C|t|}(\cN+1) (\cN/N)^{2b}
\end{split}
\end{equation}
for all $\delta > 0$, for all $t \in \bR\backslash \{ 0 \}$ and for all choices of the constant $C_b$ in the definition of $\cG_{N,t}$ (recall that $b \in \bN$ and $C_b$ enter $\cG_{N,t}$ through the definition of $\wt{\cL}_{N,t}$ in (\ref{eq:wLNt})).
\end{proposition}

As a simple corollary of Proposition \ref{lem:wG}, we can show that the expectation of the energy and the expectation and certain moments of the  number of particles operator are approximately preserved along the evolution generated by $\cG_{N,t}$; this bound will play an important role in the rest of our analysis (in particular, in the proof of Lemma \ref{lem:wPhi} below). 
\begin{corollary}\label{lem:xiN}
Assume Hypothesis A holds true. Let 
$\xi_N \in \cF_{\perp \ph_0}$ with $\| \xi_N \| \leq 1$ and such that 
\begin{equation}\label{eq:strong-xiN}
\big\langle \xi_N, \left[ \cH_N + \cN + \cN (\cN / N)^{2b} \right] \xi_N \big\rangle \leq C 
\end{equation}
uniformly in $N$ (where $b \in \bN$ is the parameter entering the definition of $\cG_{N,t}$ through (\ref{eq:wLNt})). Let $\xi_{N,t}$ be the solution of (\ref{eq:eq-xiN}) and $\xi_{2,N,t} = \cU_{2,N} (t;0) \xi_N$ with the quadratic dynamics $\cU_{2,N}$ defined in (\ref{eq:U2Nts-def}). Then, for every $b \in \bN$ and for sufficiently large $C_b > 0$, there exists a constant $C > 0$ such that 
\[ \begin{split} \big\langle \xi_{2,N,t} , \left[ \cH_N + \cN + \cN ( \cN / N)^{2b} \right] \xi_{2,N,t} \big\rangle &\leq C \exp ( C \exp (C |t|)) \\ \langle \xi_{N,t} , \left[ \cH_N + \cN + \cN ( \cN / N)^{2b} \right] \xi_{N,t} \rangle &\leq C \exp ( C \exp (C |t|)) \end{split} \]
for all $t \in \bR$.
\end{corollary}
\begin{proof}
{F}rom (\ref{eq:G2Nt-est}) and (\ref{eq:wcEbounds}) with $\delta = 1/2$
we find that, if $C_b > 0$ is large enough, 
\begin{equation}\label{eq:cGmm} \begin{split} \cG_{N,t} &\geq \eta_N (t) + \frac{1}{2} \cH_N - C e^{C|t|}  (\cN+1) + \frac{1}{4} C_b e^{C_b|t|} \cN (\cN/N)^{2b} \\
\cG_{N,t} &\leq \eta_N (t) + 2 \cH_N + C e^{C|t|}  (\cN+1) + 2 C_b e^{C_b|t|} \cN (\cN/N)^{2b} \end{split} \end{equation}
and also 
\begin{equation}\label{eq:commGN} \begin{split} i [ \cG_{N,t} , \cN ] &\leq C e^{C|t|} (\cN+1) + \cH_N + K_b e^{C|t|} \cN (\cN/N)^{2b} \\ &\leq C e^{C|t|} (\cG_{N,t} - \eta_N (t)) + C e^{C|t|} (\cN+1) \\ 
\partial_t (\cG_{N,t} - \eta_N (t)) &\leq C e^{C|t|} (\cN+1) + \cH_N + K_b e^{C|t|} \cN (\cN/N)^{2b}\\  &\leq C e^{C|t|} (\cG_{N,t} - \eta_N (t)) + C e^{C|t|} (\cN+1) \end{split} \end{equation}
We have, for any $t > 0$, 
\[\begin{split}
        &\partial_t  \langle \xi_{N,t}, (\cG_{N,t}-\eta_{N}(t) +C e^{Ct}\cN) \xi_{N,t} \rangle \\
        &\hspace{1cm}= Ce^{Ct}\langle \xi_{N,t}, i[\cG_{N,t},\cN] \xi_{N,t} \rangle + \big\langle \xi_{N,t}, \left(\partial_t \left(\cG_{N,t}-\eta_N(t)\right) +C^2e^{Ct}\cN\right) \xi_{N,t} \big\rangle. 
        \end{split}\]
Thus, from (\ref{eq:commGN}), 
\[ \begin{split} \partial_t  \langle \xi_{N,t}, &(\cG_{N,t}-\eta_N(t) +Ce^{Ct} \cN) \xi_{N,t} \rangle \\ &\leq \wt{C} \exp(\wt{C}t) \langle \xi_{N,t}, (\cG_{N,t}-\eta_N(t) + C e^{Ct}(\cN+1)) \xi_{N,t} \rangle \end{split} \]
for a sufficiently large constant $\wt{C} > 0$. Gr\"onwall's lemma yields
\[ \begin{split} \langle \xi_{N,t}, (\cG_{N,t}&-\eta_N(t) +C e^{Ct} \cN) \xi_{N,t} \rangle \\ &\leq \wt{C} \exp(\wt{C} \exp(\wt{C} t)) \langle \xi_{N}, (\cG_{N,0}-\eta_N(t) +C (\cN+1)) \xi_{N} \rangle. \end{split} \] 
{F}rom (\ref{eq:cGmm}), we conclude that, for a sufficiently large constant $C >0$,  
		\[ \begin{split} 
       \langle \xi_{N,t}, &(\cH_N +\cN+\cN(\cN/N)^{2b}) \xi_{N,t} \rangle \\
       & \le C \exp(C\exp(Ct))\langle \xi_{N}, (\cH_N +\cN+1+\cN(\cN/N)^{2b}) \xi_{N} \rangle.
        \end{split}\]
The case $t < 0$ can be treated analogously. To obtain the estimates for $\xi_{2,N,t}$ we follow exactly the same strategy, with generator $\cG_{N,t}$ replaced by $\cG_{2,N,t}$. 
\end{proof}

An important ingredient in the proof of Proposition \ref{lem:wG} is the following result, whose proof can be found, for example, in \cite{BS}; it controls the growth of moments of the number of particles operator under the action of the Bogoliubov transformation $T_{N,t}$. 
\begin{proposition} \label{lem:Bog-N} 
Assume Hypothesis A holds true and let $T_{N,t}$ denote the Bogoliubov transformation defined in (\ref{eq:Bog-trans}). Then, for every fixed $k \in \bN$ and $\delta > 0$, there exists $C > 0$ such that 
\begin{equation}\label{eq:lemBog}
\pm(T_{N,t} \cN^k T_{N,t}^*-\cN^k) \le \delta \cN^k + C.
\end{equation}
\end{proposition} 
Remark that (\ref{eq:lemBog}) requires smallness of the parameter $\ell > 0$ in (\ref{eq:Neum}) (an assumption that is included in Hypothesis A). With no assumption on the size of $\ell > 0$, (\ref{eq:lemBog}) remains true, but only for $\delta >0$ large enough. 

To show Proposition \ref{lem:wG}, we are going to consider first a simplified version of the generator $\cG_{N,t}$, given by  
\bq
\cG^\text{c}_{N,t}= (i\partial_t T_{N,t})T_{N,t}^* + T_{N,t} \cL_{N,t}^{\rm c} T_{N,t}^*. \label{def:cG}
\eq
with  $\cL_{N,t}^{\rm c}$ given by 
\begin{equation}\label{cLj}
\begin{split} 
 \cL_{N,t}^{\rm c}= & \frac{N+1}{2} \product{\varphi_{N,t}}{ [V_N (1-2f_N) \ast|\varphi_{N,t}|^2 ]\varphi_{N,t}}  - \mu_{N}(t)  \\
& + \big[\sqrt{N} a^* (Q_{N,t}[V_N \omega_N \ast|\varphi_{N,t}|^2] \varphi_{N,t}) +{\rm h.c.}\big] \\
 & +  \dGamma\Big(-\Delta + (V_Nf_N)*|\varphi_{N,t}|^2+K_{1,N,t}-\mu_{N,t} \Big) \\
   & + \Big[      \frac{1}{2} \int  \d x \d y \, K_{2,N,t}(x,y)a^*_x a^*_y + {\rm h.c.} \Big] \\
 & + \bigg[\frac{1}{\sqrt{N}}\int dx dy dx'dy'\, (Q_{N,t} \otimes Q_{N,t} V_N Q_{N,t} \otimes 1)(x,y;x',y')\varphi_{N,t} (y') a_{x}^* a_{y}^* a_{x'} \\ &\hspace{3cm} + \text{h.c.}\bigg] \\
& + \frac{1}{2N}\int dx dydx'dy' \, (Q_{N,t}\otimes Q_{N,t}V_N Q_{N,t}\otimes  Q_{N,t}) (x,y; x',y')a^*_x a^*_y a_{x'} a_{y'}.
\end{split} 
\end{equation}

The reason for considering first the generator $\cG_{N,t}^\text{c}$ is the fact that this is essentially the operator generating the fluctuation dynamics studied in \cite{BCS} for approximately coherent initial data. The only difference is the fact that, here, we always project onto the orthogonal complement of $\ph_{N,t}$. The presence of the projection $Q_t$, however, does not substantially affect the analysis of \cite{BCS}. With only small and local modifications of the proof of \cite[Theorem 3.1]{BCS}, we obtain the following proposition. 
\begin{proposition} \label{lem:cLj} 
Assume Hypothesis A holds true. Let $\cG^c_{N,t}$ be as defined in (\ref{def:cG}). Then, we have  
\bq
\cG^\text{c}_{N,t}=  \cG_{2,N,t} + \cV_N + \cE^c_{N,t}
\eq
where the quadratic generator $\cG_{2,N,t}$ is defined in \eqref{eq:cG2N} and satisfies the estimates (\ref{eq:G2Nt-est}) and where there exists a constant $C > 0$ such that the error operator 
$\cE^\text{c}_{N,t}$ satisfies  
\begin{equation}\label{eq:bndcE}
\begin{split}
\pm \cE^c_{N,t} \le \; &\delta \cV_N +N^{-\beta/2}\cK + Ce^{C|t|}\max(N^{-\alpha}, \delta^{-1})(\cN+1) \\ &+ Ce^{C|t|} \max(\delta,\delta^{-1})(\cN+1)^2/N ,\\
\pm i[\cN,\cE^c_{N,t}] \le \; &\delta \cV_N +N^{-\beta/2}\cK + Ce^{C|t|}\max(N^{-\alpha}, \delta^{-1})(\cN+1) \\
&+ Ce^{C|t|} \max(\delta,\delta^{-1})(\cN+1)^2/N ,\\
\pm \partial_t \cE^c_{N,t} \le\; &\delta \cV_N +N^{-\beta/2}\cK + Ce^{C|t|}\max(N^{-\alpha}, \delta^{-1})(\cN+1) \\
&+ Ce^{C|t|}\max(\delta,\delta^{-1})(\cN+1)^2/N 
\end{split}
\end{equation}
for all $\delta > 0$ and $t \in \bR$. 
\end{proposition}

Observe that, in \cite[Theorem 3.1]{BCS}, the operators $\cK^2$ and $\cN^2$ (the square of the kinetic energy and of the number of particles operators) are also used to control the error operator $\cE_{N,t}^c$ (see, in particular, \cite[Eq. (3.3)]{BCS}). In (\ref{eq:bndcE}), these operators do not appear; instead, we make use of the potential energy $\cV_N$ (which will be later bounded, on sectors with small number of particles, by the kinetic energy operator; see (\ref{eq:cVN<=KN})). 

Using Proposition \ref{lem:cLj}, we can proceed with the proof of Proposition \ref{lem:wG}, where we only have to control the contributions to $\cG_{N,t}$ arising from the difference $\w\cL_{N,t} - 
\cL^\text{c}_{N,t}$. 

\begin{proof}[Proof of Proposition \ref{lem:wG}.] From the definitions \eqref{def:wcG} and \eqref{def:cG} we have 
\begin{equation}\label{eq:wcENt}
\cE_{N,t}= T_{N,t}\left(\w\cL_{N,t}-\cL^{\rm c}_{N,t}\right)T_{N,t}^*-C_b e^{C_b |t|}\cN (\N/N)^{2b}+\cE^c_{N,t}
\end{equation}
We already know from Proposition \ref{lem:cLj} that $\cE^c_{N,t}$ satisfies the desired bounds. So, we focus on the first two terms on the r.h.s. 
of (\ref{eq:wcENt}). Comparing (\ref{eq:wLNt}) with (\ref{cLj}), we conclude that 
\[ T_{N,t} \left(\w\cL_{N,t}-\cL^{\rm c}_{N,t}\right) T_{N,t}^*  - C_b e^{C_b |t|} \cN (\cN/N)^{2b} = \sum_{j=1}^7 B_j \]
with
\begin{equation*}
\begin{aligned}
B_1 =  \; &\frac12\product{\varphi_{N,t}}{ [V_N\ast|\varphi_{N,t}|^2]\varphi_{N,t} } T_{N,t}\frac{\cN(\cN+1)}{N}T_{N,t}^* \\
B_2 = \; &T_{N,t} (\cL^{(1)}_{N,t} + \cL^{(3)}_{N,t}) \left(G_b(\cN/N)-1\right) T_{N,t}^* 
+\text{h.c.} \\
B_3 = \; &-T_{N,t} a^* (Q_{N,t}[V_N\ast|\varphi_{N,t}|^2] \varphi_{N,t}) \frac{\cN}{\sqrt{N}} G_b(\cN/N)T_{N,t}^*+{\rm h.c.} \\
B_4 = \; & T_{N,t}\dGamma\Big( Q_{N,t}(V_N \omega_N *|\varphi_{N,t}|^2)Q_{N,t} \Big)T_{N,t}^* \\
B_5 = \; &-  T_{N,t}\dGamma\Big( Q_{N,t}(V_N *|\varphi_{N,t}|^2)Q_{N,t} + K_{1,N,t}\Big)\frac{\cN}{N}T_{N,t}^*\\
B_6 = \; &- \frac12 \, T_{N,t} \int dxdy\; K_{2,N,t}(x,y)a^*_x a^*_y \frac{\cN}{N} T_{N,t}^*+\text{h.c.} \\
B_7 = \; &C_b e^{C_b |t|} \left( T_{N,t} \cN (\cN/N)^{2b} T_{N,t}^* -  \cN (\N/N)^{2b}\right) 
\end{aligned}
\end{equation*}
where we introduced the notation
\[ \begin{split} \cL^{(1)}_{N,t} &= \sqrt{N} a^* (Q_{N,t}[V_N \omega_N \ast|\varphi_{N,t}|^2] \varphi_{N,t})+ \text{h.c.} \\
\cL^{(3)}_{N,t} &= \frac{1}{\sqrt{N}}\int dx dy dx'dy'\, (Q_{N,t} \otimes Q_{N,t} V_N Q_{N,t} \otimes 1)(x,y;x',y')\varphi_{N,t} (y') a_{x}^* a_{y}^* a_{x'}+ \text{h.c.} \end{split} \]
Next, we control the operators $B_1, \dots , B_7$, one after the other. 

\medskip

\noindent {\it Bound for $B_1$:} From Proposition \ref{lem:Bog-N} and (\ref{eq:phi-bds}), we find immediately
$$0 \leq B_1 \leq C (\N+1)^2/N \, . $$

\noindent {\it Bound for $B_2$:} To bound the expectation of $B_2$, we write 
\begin{equation}\label{eq:B22} B_2 = \left[ T_{N,t} \cL^{(1)}_{N,t} T_{N,t}^* + T_{N,t} \cL^{(3)}_{N,t}) T^*_{N,t} \right] T_{N,t} \left(G_b(\cN/N)-1\right) T_{N,t}^*		\end{equation}
The operator in the parenthesis can be computed as in \cite[Section 3]{BCS}. The most singular contribution is the cubic term
\[ \frac{1}{\sqrt{N}} \int dx dy \, V_N (x-y)  a_x^* a_y^* a_x \ph_{N,t} (y) \]
Inserted in (\ref{eq:B22}), it produces an operator, let us denote it by $\wt{B}_2$, whose expectation can be bounded by 
\[ \begin{split} |\langle \xi , \wt{B}_2 \xi \rangle | = \; 
& \Big|\frac{1}{\sqrt{N}} \int dx dy V_N(x-y) \varphi_{N,t}(y)\langle \xi, a^*_x a^*_y a_x T_{N,t} \left(G_b(\cN/N)-1\right)T_{N,t}^* \xi\rangle \Big| \\
\leq \; &\frac{1}{\sqrt{N}}\int dxdy V_N(x-y)|\varphi_{N,t}(y)|\|a_x a_y \xi\|\|a_x T_{N,t} \left(G_b(\cN/N)-1\right)T_{N,t}^* \xi\| \\
\leq \; &\frac{\delta}{2N}\int dxdy V_N(x-y)\|a_x a_y \xi\|^2 \\ &+C\delta^{-1}e^{C|t|}\int dxdy V_N(x-y)\|a_x T_{N,t} \left(G_b(\cN/N)-1\right)T_{N,t}^* \xi\|^2 \\
\leq \; &\delta\langle \xi, \cV_N \xi\rangle + K_b \delta^{-1}e^{C|t|}\, \big\langle \xi, (\cN+1) \left[ ((\cN+1)/N)^2 + \left((\cN+1)/N\right)^{2b} \right] \xi \big\rangle \end{split}\]
for any $\delta > 0$ and for an appropriate constant $K_b$ depending on the choice of $b$. Here we used Proposition \ref{lem:Bog-N}. Other terms contributing to $B_2$ can be bounded in a similar fashion. We conclude that 
\[ \pm B_2 \leq \delta \cV_N + K_b \delta^{-1} e^{C|t|} (\cN+1) \left[ ((\cN+1)/N)^{2} + ((\cN+1)/N)^{2b} \right] \]

\noindent {\it Bound for $B_3$:} Let us now deal with $B_3$. Since $\| Q_{N,t} [V_N * |\ph_{N,t}|^2]\ph_{N,t} \| \leq C \exp (C|t|)$, we obtain, with Cauchy-Schwarz,   
\[ \pm B_3 \leq K_b \delta e^{C|t|} \frac{(\cN+1)^2}{N} + K_b e^{C|t|} \delta^{-1} \cN + K_b e^{C|t|} \delta^{-1} (\cN+1) (\cN/N)^{2b}  \]
for every $\delta > 0$ and for an appropriate constant $K_b > 0$ depending on $b \in \bN$. 

\noindent {\it Bound for $B_4$:} From (\ref{eq:fN-bound}), we have
\[ \| Q_{N,t} (V_N \omega_N * |\ph_{N,t}|^2) Q_{N,t} \|_\infty 
\leq C N^{\beta-1} e^{C|t|} \]
Hence, with Proposition \ref{lem:Bog-N}, we find 
\[ \pm B_4 \leq C N^{\beta-1} \leq C N^{\beta-1} (\cN + 1) \]

\noindent {\it Bound for $B_5$:} Similarly, since $\| K_{1,N,t} \| = \| Q_{N,t} \wt{K}_{1,N,t} Q_{N,t} \| \leq \| \wt{K}_{1,N,t} \|$, 
\[ 
\begin{split}
\|\w K_{1,N,t}\| = \; & \sup_{\|f\|_{L^2}=1}\left|\int \overline{f(x)}\varphi_{N,t}(x) V_N(x-y)\overline{\varphi_{N,t}(y)}f(y)\d x\d y \right| \\
\leq \; &\sup_{\|f\|_{L^2}=1}\|\varphi_{N,t}\|^2_{L^\infty}\int \frac{|f(x)|^2+|f(y)|^2}{2}V_N(x-y)\d x\d y \leq \; C e^{C|t|}
\end{split} 
\]
and $\| Q_{N,t} (V_N * |\ph_{N,t}|^2) Q_{N,t} \|_\infty \leq C \exp (C|t|)$, we obtain with Proposition \ref{lem:Bog-N} that 
\[ \pm B_5 \leq C e^{C|t|} \frac{(\cN+1)^2}{N}. \]

\noindent {\it Bound for $B_6$:} Proceeding as in \cite[Prop. 3.5]{BCS} we find
\begin{equation}\label{eq:B6} \begin{split} 
B_6 = \; &- \frac{1}{2N} \int dxdy\; K_{2,N,t}(x;y)\langle \text{sh}_x, \text{ch}_y\rangle T_{N,t} \cN T_{N,t}^* \\ &+ \frac{1}{2N} \int dxdy\; V_N(x-y)\varphi_{N,t}(x)\varphi_{N,t}(y)a^*_x a^*_y T_{N,t} \cN T_{N,t}^* \\ &+ \frac{1}{N} \, \cE_{6,N} T^*_{N,t} \cN T_{N,t} + \text{h.c.} \end{split} \end{equation}
where the operator $\cE_{6,N}$ is such that 
\begin{equation}\label{eq:E6} \cE_{6,N}^2 \leq C e^{C|t|} (\cN+1)^2. \end{equation}
Since 
\[ \Big|\int dxdy\; K_{2,N,t}(x;y)\langle \text{sh}_x, \text{ch}_y\rangle \Big|\leq Ce^{C|t|} \, , \]
the expectation of the first term on the r.h.s. of (\ref{eq:B6}) is bounded, with Proposition \ref{lem:Bog-N}, by 
\[ \Big| \frac{1}{2N} \int dxdy\; K_{2,N,t}(x;y)\langle \text{sh}_x, \text{ch}_y  \rangle \, \langle \xi , T_{N,t} \cN T_{N,t}^* \xi \rangle \Big| \leq C N^{-1} \langle \xi, (\cN+1) \xi \rangle \]
The expectation of the second term on the r.h.s. of (\ref{eq:B6}) can be controlled by 
\begin{equation*}
\begin{aligned}
\Big|\frac{1}{2N} \int & dxdy\; V_N(x-y)\varphi_{N,t}(x)\varphi_{N,t}(y)\, \langle \xi, a^*_x a^*_y T_{N,t} \cN  T_{N,t}^*\xi\rangle\Big| \\
&\leq \frac{1}{2N}\int dxdy\; V_N(x-y)|\varphi_{N,t}(x)| \, |\varphi_{N,t}(y)| \| a_x a_y \xi\|\, \|T_{N,t}\cN T_{N,t}^*\xi\| \\
& \leq \frac{1}{2N}\int dxdy\; V_N(x-y)\left[ \delta \| a_x a_y \xi\|^2 +\delta^{-1} |\varphi_{N,t}(x)|^2 \, |\varphi_{N,t}(y)|^2 \| \cN T_{N,t}^*\xi\|\right] \\
&\leq \delta \langle \xi, \cV_N \xi\rangle + C\delta^{-1} N^{-1} e^{C|t|}\langle \xi, (\cN + 1)^2 \xi\rangle
\end{aligned}
\end{equation*}
where we used once again Proposition \ref{lem:Bog-N}. As for the last term on the r.h.s. of (\ref{eq:B6}), it can be estimated using (\ref{eq:E6}) and Proposition \ref{lem:Bog-N}. We conclude that 
\[ \pm B_6 \leq \delta \cV_N + C \delta^{-1} e^{C|t|} \frac{(\cN+1)^2}{N} \]
for any $\delta > 0$. 

\noindent {\it Bound for $B_7$:} with Proposition \ref{lem:Bog-N} we find 
\[ \pm B_7 \leq \frac{1}{2} C_b e^{C_b |t|} (\cN+1) ((\cN+1)/N))^{2b} \]
if $\ell > 0$ in (\ref{eq:Neum}) is chosen sufficiently small. 

Combining all these bounds with the bounds (\ref{eq:bndcE}) for the error term $\cE_{N,t}^c$, we obtain the first estimate in (\ref{eq:wcEbounds}) for the error term $\cE_{N,t}$. 

The bound for the commutator $i[\cN , \cE_{N,t}]$ follows from the observation that the commutator with $\cN$ of every monomial $A$ in 
creation and annihilation operators appearing in $\cE_{N,t}$ is 
given by $\lambda A$, where $\lambda \in \{ 0 , \pm 1, \pm 2 , \pm 3\}$. Hence, $[i\cN, \cE_{N,t}]$ can be bounded exactly like we did for $\cE_{N,t}$. 

Similarly, the bound for the time-derivative $\partial_t \cE_{N,t}$ is established by noticing that the time derivative of every monomial $A$ contributing to $\cE_{N,t}$ is the sum of finitely many terms having again the same form of $A$, just with one factor $\ph_{N,t}$ replaced by the time derivative $\partial_t \ph_{N,t}$ (the generator $\cG_{N,t}$ only depends on time through the solution $\ph_{N,t}$ of the nonlinear Hartree equation (\ref{eq:NLSN})). Therefore, to bound $\partial_t \cE_{N,t}$ we proceed exactly as we did for $\cE_{N,t}$, with the only difference that, sometimes, we have to use the bound for $\partial_t \ph_{N,t}$ in (\ref{eq:phi-bds}) rather than the corresponding bound for $\ph_{N,t}$. 
\end{proof}

With Proposition \ref{lem:wG}, we are now ready to prove Lemma \ref{lem:Compare-xiN}. 

\begin{proof}[Proof of Lemma \ref{lem:Compare-xiN}] 
Let $\alpha= \min(\beta,1-\beta)/2$ and $M=N^\alpha$. We have  
$$
\| \w\xi_{N,t}-\xi_{2,N,t}\|^2 = 2 \left[ 1-\Re \langle \xi_{N,t},\xi_{2,N,t}\rangle \right] 
$$
and we decompose, with $M/2\le m\le M$, 
\[ \begin{split} 
\langle \xi_{N,t}, \xi_{2,N,t} \rangle &= \langle \xi_{N,t}, \1^{\le m}\xi_{2,N,t} \rangle +  \langle \xi_{N,t}, \1^{>m} \xi_{2,N,t} \rangle \\ 
&= \frac{2}{M} \sum_{m= M/2+1}^M \left[ \langle \xi_{N,t}, \1^{\le m}\xi_{2,N,t} \rangle +  \langle \xi_{N,t}, \1^{>m} \xi_{2,N,t} \rangle \right] .
\end{split} \]
where we used the notation $\1^{\le m}=\1(\cN\le m)$ and $\1^{>m}=\1-\1^{\le m}$. 

\noindent  {\it Many-particle sectors.} From Cauchy-Schwarz and the bounds in Corollary \ref{lem:xiN}, we find 
\begin{equation*}
\begin{split}
|\langle \xi_{N,t}, \1^{>m} \xi_{2,N,t} \rangle| &\le \|\1^{>m} \xi_{N,t}\|. \|\1^{>m}\xi_{2,N,t}\| \\
&\le \langle \xi_{N,t}, (\cN/m) \xi_{N,t}\rangle^{1/2}  \langle \xi_{2,N,t}, (\cN/m) \xi_{2,N,t}\rangle^{1/2} \\
&\le C M^{-1} \exp(C\exp(C|t|)). 
\end{split} 
\end{equation*}
for a constant $C > 0$ depending on $b$. Averaging over $m\in [M/2+1,M]$, we conclude that 
\begin{align} \label{eq:MP-2a}
\Big| \frac{2}{M} \sum_{m=M/2+1}^M \langle \xi_{N,t}, \1^{>m} \xi_{2,N,t} \rangle \Big| \le C N^{-\alpha} \exp(C\exp(C|t|)).
\end{align}

\noindent  {\it Few-particle sectors.} From the Schr\"odinger equations for $\xi_{N,t}$ and $\xi_{2,N,t}$, we find  
\begin{align*}
\Re \frac{d}{dt} \langle \xi_{N,t}, \1^{\le m}\xi_{2,N,t} \rangle &=\Im \Big\langle \xi_{N,t}, \Big[ (\cG_{N,t}-\cG_{2,N,t}) \1^{\le m} +[\cG_{2,N,t},\1^{\le m}] \Big] \xi_{2,N,t} \Big\rangle .
\end{align*}
Using Proposition \ref{lem:wG}, in particular (\ref{eq:wcEbounds}) with $\delta=N^{\alpha}$, we obtain
\[ \begin{split} 
\pm (&\cG_{N,t} -\cG_{2,N,t}) \\ &\le  \Big[ N^{\alpha} \cV_N + N^{\alpha} (\cN+1)^2/N + (\cN+1)(\cN/N)^{2b}+N^{-\alpha}(\cK+\cN+1)
\Big] C \exp(Ct) 
\end{split} \]
for a constant $C > 0$ depending on $b$. We choose $b \in \bN$ large enough so that $2b(\alpha-1) < -\alpha$ (i.e. $b > \alpha/
(2(1-\alpha))$). Then, using the simple operator estimate
\bq \label{eq:cVN<=KN}
0\le \cV_N \le CN^{\beta-1} \cK \cN 
\eq
which follows by quantization of the two-body estimate $V_N(x-y)\le C N^{\beta} (-\Delta_x-\Delta_y)$, projecting to the sector with 
$\cN\le m+2$ (where $m\le N^\alpha$), and using also the inequality $2\alpha -1 < -\alpha$ (since, by definition, $\alpha < 1/4$) we find 
\bq \label{eq:Gn-G2n-rstM}
\pm \1^{\le m+2}(\cG_{N,t}-\cG_{2,N,t}) \1^{\le m+2} \le C N^{-\alpha} (\cK+\cN+1) \exp(C|t|).
\eq
Since $\cG_{N,t}-\cG_{2,N,t}$ contains terms with at most two creation operators, we have the obvious identity
$$(\cG_{N,t}-\cG_{2,N,t})\1^{\le m}=\1^{\le m+2}(\cG_{N,t}-\cG_{2,N,t})\1^{\le m+2}\1^{\le m}.$$
From \eqref{eq:Gn-G2n-rstM} we find, by Cauchy-Schwarz, 
\begin{equation}
\label{eq:form}
\begin{split} 
&|\langle \xi_{N,t},(\cG_{N,t}-\cG_{2,N,t}) \1^{\le m}  \xi_{2,N,t} \rangle| \\
&=|\langle \xi_{N,t}, \1^{\le m+2} (\cG_{N,t}-\cG_{2,N,t}) \1^{\le m+2} \1^{\le m}  \xi_{2,N,t} \rangle| \\
&\le C N^{-\alpha} \exp(C|t|)\langle \xi_{N,t}, (\cK+\cN+1) \xi_{N,t}\rangle^{1/2} \langle \1^{\le m}  \xi_{2,N,t}, (\cK+\cN+1) \xi_{2,N,t}\rangle^{1/2}.
\end{split}
\end{equation}
Inserting the energy estimates in Corollary \ref{lem:xiN}, we find that
$$
|\langle \xi_{N,t},(\cG_{N,t}-\cG_{2,N,t}) \1^{\le m}  \xi_{2,N,t} \rangle| \le C N^{-\alpha} \exp(\exp(C|t|)).
$$
In (\ref{eq:form}), we used the fact that, if $D$ is a self-adjoint and $F$ a non-negative operator on a Hilbert space $\frak{h}$ with $\pm D \leq F$ then, for every $\phi,\psi \in \frak{h}$, we have (using the fact that $D+F \geq 0$) 
\[ \begin{split} |\langle \phi, D \psi \rangle| &\leq |\langle \phi , (D+F) \psi \rangle| + |\langle \phi ,   F \psi \rangle| \\ &\leq \kappa \langle \phi, (D+F) \phi \rangle + \kappa^{-1} \langle \psi, (D+F) \psi \rangle + \kappa \langle \phi, F \phi \rangle + \kappa^{-1} \langle \psi, F \psi \rangle \\ &\leq 3 \kappa \langle \phi, F \phi \rangle + 3 \kappa^{-1} \langle \psi, F \psi \rangle \end{split} \]
for every $\kappa > 0$. With $\kappa = \langle \psi, F \psi \rangle^{1/2} \langle \phi, F \phi \rangle^{-1/2}$, we find 
\[   |\langle \phi, D \psi \rangle| \leq 6 \langle \phi, F \phi \rangle^{1/2} \langle \psi, F \psi \rangle^{1/2} \]

Next, we turn to the commutator $[\cG_{2,N,t},\1^{\le m}]$. We observe that 
\begin{equation}\label{eq:G1comm}
[\cG_{2,N,t},\1^{\le m}] = \1^{> m} \cG_{2,N,t} \1^{\le m} - \1^{\le m} \cG_{2,N,t} \1^{>m}. 
\end{equation}
Consider the first term on the r.h.s. of (\ref{eq:G1comm}). Only terms in $\cG_{2,N,t}$ with two creation operators give a non-vanishing contribution; hence, 
\[ \begin{split} \langle \xi_1, \1^{> m} &\cG_{2,N,t} \1^{\le m} \xi_2 \rangle \\ &= 
\big\langle \xi_1, \left[ \chi (\cN = m+2) \cG_{2,N,t} \chi (\cN = m) + \chi (\cN = m+1) \cG_{2,N,t} \chi (\cN = m) \right] \xi_2 \big\rangle \end{split} \]
Estimating terms in $\cG_{2,N,t}$ with two creation operators similarly as in Proposition \ref{lem:cLj}, we obtain
\begin{align*}
& \Big|\frac{2}{M} \sum_{m=M/2+1}^M \langle \xi_{N,t},  i[\cG_{2,N,t},\1^{\le m}] \xi_{2,N,t}\rangle \Big| \\
&\le C M^{-1} \exp(C|t|)  \langle \xi_{N,t}, (\cN+1) \xi_{N,t}\rangle^{1/2}\langle \xi_{2,N,t}, (\cN+1) \xi_{2,N,t}\rangle^{1/2}\\
&\le C N^{-\alpha} \exp(C\exp(C|t|)). 
\end{align*}
where we used Corollary \ref{lem:xiN} and the choice $M=N^{\alpha}$. In summary, we have proved that
$$
\Big| \frac{2}{M} \sum_{m=M/2+1}^M \frac{d}{dt} \, \langle\xi_{N,t}, \1^{\le m}\xi_{2,N,t}\rangle \Big| \le C N^{-\alpha}\exp(C\exp(C|t|)).
$$
Consequently,
\begin{align*}
&\Re \frac{2}{M} \sum_{m=M/2+1}^M \langle\xi_{N,t}, \1^{\le m}\xi_{2,N,t}\rangle \\
&\ge \Re \frac{2}{M}  \sum_{m=M/2+1}^M \langle\xi_{N,0}, \1^{\le m}\xi_{2,N,0}\rangle - C N^{-\alpha}\exp(C\exp(C|t|)).
\end{align*}
With the assumption (\ref{eq:ass-xiNb}) on the initial datum $\xi_{N,0}=\xi_{2,N,0}=\xi_N$, we find 
\begin{align*}
\langle \xi_{N,0}, \1^{\le m}\xi_{2,N,0}\rangle &= \|\1^{\le m} \xi_N\|^2 = 1- \|\1^{>m} \xi_N\|^2\\
&\ge 1- \langle \xi_N, (\cN/m) \xi_N\rangle \ge 1- CM^{-1}= 1-CN^{-\alpha}.
\end{align*}
Thus
\begin{align*}
\Re \frac{2}{M}  \sum_{m=M/2+1}^M \langle\xi_{N,t}, \1^{\le m}\xi_{2,N,t}\rangle \ge 1- C N^{-\alpha}\exp(C\exp(C|t|)).
\end{align*}

Combining the latter bound with \eqref{eq:MP-2a},  we arrive at
$$
\Re \langle\xi_{N,t}, \xi_{2,N,t}\rangle \ge 1- C N^{-\alpha}\exp(C\exp(C|t|)).$$
We conclude that 
$$
\| \xi_{N,t} - \xi_{2,N,t}\|^2 \le 2(1-\Re \langle \xi_{N,t}, \xi_{2,N,t}\rangle )\le C N^{-\alpha} \exp(C\exp(C|t|)). 
$$
\end{proof}

The localization argument used in the above proof is similar to that in \cite{NN2,NN3}. The main idea is to employ the operator inequality \eqref{eq:cVN<=KN} in the sector of few particles. This argument will be used again below.

\section{Approximation of fluctuation dynamics}
\label{sec:phiwphi}

In this section, we show Lemma \ref{lem:Phi-wPhi}. To this end, we will make use of the following energy estimates. 

\begin{lemma} \label{lem:Phi} 
Assume Hypothesis A holds true. Let $\xi_N \in \cF_{\perp}$ with $\| \xi_N \| \leq 1$ and 
\begin{equation}\label{eq:lmPhi-xiN} \langle \xi_N, (\cH_N +\cN + \cN^2/N) \xi_N\rangle \leq C, \end{equation}
uniformly in $N$. Let $\Phi_{N,t}$ be as defined in \eqref{eq:Phi}. Then there exists a constant $C > 0$ such that 
\begin{equation}\label{eq:lmPhi} 
\langle \Phi_{N,t}, (\cH_N +\cN) \Phi_{N,t} \rangle \le C N^{\beta} \,  \exp(C\exp(C|t|))
\end{equation}
for all $t \in \bR$. 
\end{lemma} 

\begin{proof}
We recall that $\Phi_{N,t}$ solves the Schr\"odinger equation (\ref{eq:Phi}) with the generator (\ref{eq:cLNt}) that can be decomposed into 
$$\cL_{N,t} =C_{N,t}+\cH_{N,t}+\cR_{N,t}$$
with the constant part 
 \bq
C_{N,t}= \frac{N+1}{2} \product{\varphi_{N,t}}{ [V_N (1-2f_N) \ast|\varphi_{N,t}|^2 ]\varphi_{N,t}} - \mu_{N,t}, \label{def:C_Nt}
 \eq 
the projected Hamilton operator
\[ \begin{split} \cH_{N,t} = \; &d\Gamma (-\Delta) \\ &+ \frac{1}{2N}  \int dx dy dx' dy' \, \left[ (Q_{N,t} \otimes Q_{N,t}) V_N (Q_{N,t} \otimes Q_{N,t})\right] (x,y ; x' , y') a_x^* a_y^* a_{x'} a_{y'} \end{split} \]
and the rest $$\cR_{N,t}=\sum_{i=1}^7 \cR_{N,t}^i$$ where 
\begin{equation}\label{eq:RN7}
\begin{split} 
\cR_{N,t}^1 &= \frac12\product{\varphi_{N,t}}{ [V_N\ast|\varphi_{N,t}|^2]\varphi_{N,t} }  \frac{\cN(\cN+1)}{N}  \\
\cR_{N,t}^2 &=  \Big[ a^* (Q_{N,t}[(V_N \omega_N) \ast|\varphi_{N,t}|^2] \varphi_{N,t}) - a^* (Q_{N,t}[V_N\ast|\varphi_{N,t}|^2] \varphi_{N,t}) \frac{\cN}{N}  \Big] \sqrt{N-\cN} \\ &\hspace{.5cm} +{\rm h.c.}  \\
\cR_{N,t}^3 &=   \dGamma\Big( (V_Nf_N)*|\varphi_{N,t}|^2 + K_{1,N,t} -\mu_{N,t} \Big) +\dGamma\Big( Q_{N,t}(V_N \omega_N *|\varphi_{N,t}|^2)Q_{N,t} \Big)\\
\cR_{N,t}^4 &=    - \dGamma\Big( Q_{N,t}(V_N *|\varphi_{N,t}|^2)Q_{N,t} + K_{1,N,t}\Big)\frac{\cN}{N}\\
\cR_{N,t}^5   &=  \frac{1}{2} \int  \d x \d y \, K_{2,N,t}(x,y)a^*_x a^*_y  + {\rm h.c.} \\
\cR_{N,t}^6   &=  \frac{1}{2} \int  \d x \d y \, K_{2,N,t}(x,y)a^*_x a^*_y \left(\frac{\sqrt{(N-\cN)(N-\cN-1)}}{N}-1\right) + {\rm h.c.} \\
\cR_{N,t}^7   &= \frac{1}{\sqrt{N}}\int \d x \d y \d x' \d y' \, (Q_{N,t} \otimes Q_{N,t} V_N Q_{N,t} \otimes 1)(x,y;x',y') \\ &\hspace{6cm} \times a_x^* a_y^* a_{x'} \varphi_{N,t} (y') \sqrt{\frac{N-\cN}{N}} + \text{h.c.}
 \end{split}\end{equation} 
The proof of Lemma \ref{lem:Phi} is divided into three steps. In the first step, we bound the rest operator $\cR_{N,t}$, its commutator with $\cN$ and its time derivative, through the number of particles operator $\cN$ and the Hamiltonian $\cH_N$. In the second step we use these bounds and, with Gr\"onwall's Lemma, we control the expectation on the r.h.s. of (\ref{eq:lmPhi}) in terms of its initial value at time $t=0$. Finally, in the third step, we control the expectation of $\cH_N$ and $\cN$ in the initial state $\Phi_{N,0} = T_{N,0} \xi_N$ through the expectation of the same operators in the state $\xi_N$, making use of the assumption (\ref{eq:lmPhi-xiN}).

\medskip

\noindent \textit{Step 1.} We claim that, for all $\delta > 0$ there exists $C > 0$ with  
\begin{equation} \label{eq:gronwallestimatesPhi}
\begin{aligned}
\pm \cR_{N,t} &\leq \delta \cV_N + C e^{C |t|} (\cN + N^\beta)  \\
\pm i[\cR_{N,t},\cN] &\leq \delta \cV_N + C_{\epsilon}e^{Ct}(\cN +N^{\beta}) \\
\pm \partial_t \cR_{N,t} &\leq \delta \cV_N + C_{\epsilon}e^{Ct}(\cN +N^{\beta}).
\end{aligned}
\end{equation}
as operator inequality on $\cF_{\perp \ph_{N,t}}^{\leq N}$. We will focus on the proof of the bound for $\cR_{N,t}$. The other two estimates in (\ref{eq:gronwallestimatesPhi}) can be shown similarly, since the commutator $i[\cR_{N,t} , \cN]$ and the derivative $\partial_t \cR_{N,t}$ contain the same terms appearing in $\cR_{N,t}$, multiplied by a constant in $\{ 0 , \pm 1, \pm 2 \}$ in the first case 
and with a factor $\ph_{N,t}$ replaced by its derivative $\partial_t \ph_{N,t}$ in the second case. We follow here \cite[Theorem 3]{NN2}, where more details can be found. 

\medskip

\noindent {\it Step 1.1}: Since
$$\product{\varphi_{N,t}}{ [V_N \ast|\varphi_{N,t}|^2 ]\varphi_{N,t}} \leq \|V_N\|_{L^1}\|\varphi_{N,t}\|_{L^4}^4 \leq C $$
and $\cN/N\leq 1$ on the truncated Fock space $\cF_{\perp \ph_{N,t}}^{\leq N}$, we have 
\begin{equation*}
0\leq \cR_{N,t}^1  \leq C \cN. 
\end{equation*}
\medskip

\noindent {\it Step 1.2}: We divide $\cR_{N,t}^2 = \cR_{N,t}^{2,1} + \cR_{N,t}^{2,2}$ with 
\begin{equation}\label{eq:R2122} \begin{split} \cR_{N,t}^{2,1} &=\Big[ a^* (Q_{N,t}[(V_N \omega_N) \ast|\varphi_{N,t}|^2] \varphi_{N,t}) \Big] \sqrt{N-\cN} + \text{h.c.}, \\ \cR_{N,t}^{2,2} &= \sqrt{N}\Big[ a^* (Q_{N,t}[V_N\ast|\varphi_{N,t}|^2] \varphi_{N,t}) \frac{\cN}{N}  \Big] \sqrt{\frac{N-\cN}{N}} + \text{h.c.} \end{split} \end{equation}
Using the Cauchy--Schwarz inequality, we find, for arbitrary $\xi \in \cF^{\leq N}_{\perp \ph_{N,t}}$, 
\begin{equation*}
\begin{aligned}
&\left| \langle \xi, \cR_{N,t}^{2,1} \xi \rangle \right| \leq \sqrt{N}\Big\| Q_{N,t} [(V_N\omega_N)\ast|\varphi_{N,t}|^2] \varphi_{N,t} \Big\|_{2}\|\cN^{1/2}\xi\|  \| \xi \| 
\end{aligned}
\end{equation*}
Since, with (\ref{eq:fN-bound}), 
\begin{align*}
\Big\| Q_{N,t}[(V_N\omega_N)\ast|\varphi_{N,t}|^2] \varphi_{N,t} \Big\|_{L^2}  \le \Big\| [(V_N\omega_N)\ast|\varphi_{N,t}|^2] \varphi_{N,t} \Big\|_{L^2} \leq C N^{\beta-1} e^{C|t|} 
\end{align*}
we conclude that 
$$\pm \cR_{N,t}^{2,1}   \leq Ce^{C|t|}\left(N^{2\beta-1} +\cN\right) \leq C e^{C|t|} (N^\beta + \cN) .$$
As for the second term in (\ref{eq:R2122}), using 
$$\Big\|[V_N\ast|\varphi_{N,t}|^2] \varphi_{N,t} \Big\|_{L^2}  \le C \, , $$
the Cauchy--Schwarz inequality and the fact that, on $\cF^{\leq N}_{\perp \ph_{N,t}}$, $\cN/N\leq 1$, we find hat
\begin{equation*}
 \pm \cR_{N,t}^{2,2}\leq Ce^{C|t|} \cN.
\end{equation*}

\medskip

\noindent {\it Step 1.3}: Recall that for an operator $B$ on $L^2(\R^3)$ we have $\pm d\Gamma (B)\leq \|B\| \cN$. Since
\begin{equation*}
\begin{split}
\|(V_N f_N)*|\varphi_{N,t}|^2 \|_{L^\infty} &\leq \|\varphi_{N,t}\|^2_{L^\infty} \|V_Nf_N\|_{L^1}\leq  Ce^{Ct} \\
|\mu_{N,t}| &=\big|\product{\varphi_{N,t}}{ [(V_N \omega_{N}) \ast|\varphi_{N,t}|^2 ]\varphi_{N,t}}\big| \\ &\leq C N^{\beta-1} e^{C|t|}  \\
\|Q_{N,t}(V_N\omega_N)*|\varphi_{N,t}|^2Q_{N,t}\| &\leq \|(V_N\omega_N)*|\varphi_{N,t}|^2\|_{L^\infty}\\ &\leq C N^{\beta-1}e^{C|t|} 
\end{split}
\end{equation*}
and  
\begin{equation}
\begin{split}
\|K_{1,N,t}\| &= \|Q_{N,t}\w K_{1,N,t}Q_{N,t}\| \leq \|\w K_{1,N,t}
\| \\ &=\sup_{\|f\|_{L^2}=1}\left|\int \overline{f(x)}\varphi_{N,t}(x) V_N(x-y)\overline{\varphi_{N,t}(y)}f(y)\d x\d y \right|\\ &\leq \frac{\|\varphi_{N,t}\|^2_{L^\infty}}{2} 
\sup_{\|f\|_{L^2}=1} \int (|f(x)|^2+|f(y)|^2) V_N(x-y) dx dy \leq Ce^{C|t|}
\end{split}  \label{eq:boundK1}
\end{equation}
we conclude that 
$$\pm \cR_{N,t}^{3}\leq Ce^{C|t|}\cN.$$

\medskip

\noindent
{\it Step 1.4}: Proceeding similarly to Step 3 and using the fact that $d\Gamma (B)$ commutes with $\cN$, we find 
\[ \pm \cR_{N,t}^{4}\leq Ce^{C|t|}\cN. \]

\medskip

\noindent {\it Step 1.5}: To bound the term $\cR^5_{N,t}$ we observe that, for any $\delta > 0$,   
\begin{equation*}\begin{split}
\delta  \dGamma (1-\Delta) \pm \Big[\frac{1}{2} &\int  \d x \d y \, K_{2,N,t}(x,y) a^*_x a^*_y + {\rm h.c.} \Big] \\ &\geq - \frac{1}{2\delta} \left\| (1-\Delta)^{-1/2} K^*_{2,N,t} \right\|_\text{HS}^2 \geq - \frac{1}{2\delta} \left\| (1-\Delta)^{-1/2} \wt{K}^*_{2,N,t} \right\|_\text{HS}^2
\end{split}
\end{equation*}
from \cite[Lemma 9]{NamNapSol-16}. Since $\wt{K}_{2,N,t} (x;y) = V_N (x-y) \ph_{N,t} (x) \ph_{N,t} (y)$, we find
\[ \begin{split} 
\Big\| (1- &\Delta)^{-1/2} \wt{K}^*_{2,N,t} \Big\|_\text{HS}^2 \\ &= \tr \; \w{K}_{2,N,t} (1-\Delta)^{-1} \w{K}^*_{2,N,t} \\ &= C\int dx dy dz \, V_N (x-y)  \, \frac{e^{-|y-z|}}{|y-z|} \, V_N (z-x) \, |\ph_{N,t} (x)|^2 \ph_{N,t} (y) \ph_{N,t} (z) \\ &\leq \| \ph_{N,t} \|^2_\infty \| \ph_{N,t} \|_2^2 \int V_N (z) \left[ V_N * \frac{1}{|.|} \right] (z) dz \\ &\leq C e^{C|t|} \int \frac{|\widehat{V}_N (p)|^2}{p^2} dp =  C e^{C|t|} \int \frac{| \widehat{V} (p/N^\beta)|^2}{p^2} dp\leq C N^\beta e^{C|t|} \int \frac{|\widehat{V} (p)|^2}{p^2} dp \\ &\leq C N^\beta e^{C|t|}  \end{split} \]
We obtain that, for any $\delta > 0$, 
\[ \pm \cR_{N,t}^5 = \pm \Big[ \frac{1}{2} \int  dx dy \, K_{2,N,t}(x,y) a^*_x a^*_y + {\rm h.c.} \Big] \leq \delta d\Gamma (1-\Delta) + C \delta^{-1} N^\beta e^{C|t|} \] 

\medskip

\noindent
{\it Step 1.6}: To bound $\cR_{N,t}^6$, we observe that, by Cauchy-Schwarz, we have 
\[ \begin{split} | \langle \xi , \cR_{N,t}^6 \xi \rangle | & \leq C \int dx dy \, V_N (x-y) |\ph_{N,t} (x)| | \ph_{N,t} (y)| \| a_x a_y \xi \| \\ &\hspace{6cm} \times \left\| \left( \frac{\sqrt{N-\cN)(N-\cN-1)}}{N} - 1 \right) \xi \right\| 
\\ &\leq \frac{C}{\sqrt{N}} \int dx dy \, V_N (x-y) |\ph_{N,t} (x)| | \ph_{N,t} (y)| \| a_x a_y \xi \| \, \| (\cN+1)^{1/2} \xi\|
\\ &\leq \delta \langle \xi, \cV_N \xi \rangle  + C \delta^{-1} \| (\cN+1)^{1/2} \xi \|^2
\end{split} \]
which implies that
\[ \pm \cR_{N,t}^6 \leq \delta \cV_N + C \delta^{-1} (\cN + 1) \]

\medskip

\noindent
{\it Step 1.7}: For $\xi \in \cF^{\leq N}_{\perp \ph_{N,t}}$, we have, using Cauchy-Schwarz inequality, 
\[ \begin{split}
\Big|\langle \xi, \cR_{N,t}^7 \xi \rangle \Big| &\leq \frac{1}{\sqrt{N}} \int V_N (x-y) |\ph_{N,t} (y)| \, \| a_x a_y \xi \| \, \left\| a_x \sqrt{\frac{N-\cN}{N}} \xi \right\| dx dy \\  &\leq \delta \langle \xi, \cV_N \xi \rangle + C \| \ph_{N,t} \|_\infty^2 \langle \xi , \cN \xi \rangle 
\end{split} \]
and therefore
\[ \pm \cR_{N,t}^7 \leq \delta \cV_N + C \delta^{-1} e^{C|t|} \cN \]

Combining the results of Step 1.1 - Step 1.7, we obtain (\ref{eq:gronwallestimatesPhi}). 

\medskip

\noindent \textit{Step 2.} There exists a constant $C > 0$ such that 
\begin{equation}\label{eq:step2-cl} \langle \Phi_{N,t}, (\cH_N +\cN) \Phi_{N,t} \rangle \le C \exp(C \exp(C|t|)) \langle \Phi_{N,0}, (\cH_N +\cN+N^{\beta}) \Phi_{N,0} \rangle \end{equation}
for all $t \in \bR$.

We focus on $t > 0$ (the case $t < 0$ can be handled similarly). 
We have  
\begin{equation*}
\begin{aligned}
\partial_t \big\langle &\Phi_{N,t},\big( \cL_{N,t} - C_{N,t}+ C e^{Ct}(\cN+N^{\beta})\big)\Phi_{N,t}\big\rangle \\
& =Ce^{Ct} \langle \Phi_{N,t}, i[\cR_{N,t},\cN]\Phi_{N,t}\rangle +  
\langle \Phi_{N,t},(\partial_t  \cR_{N,t}+C^2e^{Ct}(\cN+N^{\beta}))\Phi_{N,t}\rangle. 
\end{aligned}
\end{equation*}
The second and third bound in \eqref{eq:gronwallestimatesPhi} 
imply that there exists a constant $\wt{C}$ such that 
\begin{equation*}
\begin{split}
\partial_t \langle \Phi_{N,t},& ( \cL_{N,t}-C_{N,t}+Ce^{Ct}(\cN+N^{\beta}))\Phi_{N,t}\rangle \\
& \leq \wt{C} e^{\wt{C} t} \langle \Phi_{N,t},( \cL_{N,t}-C_{N,t}+Ce^{Ct}(\cN+N^{\beta}))\Phi_{N,t}\rangle.
\end{split}
\end{equation*}
Gr\"onwall's Lemma gives 
\begin{equation*}
\begin{split}
\langle \Phi_{N,t}, ( \cL_{N,t}&-C_{N,t}+Ce^{Ct}(\cN+N^{\beta}))\Phi_{N,t}\rangle \\
& \leq \wt{C}\exp(\wt{C}\exp(\wt{C}t)) \langle \Phi_{N,0},( \cL_{N,0}-C_{N,0}+\cN+N^{\beta})\Phi_{N,0}\rangle
\end{split}
\end{equation*}
The first inequality in \eqref{eq:gronwallestimatesPhi} implies (\ref{eq:step2-cl}).

\medskip

\noindent \textit{Step 3.} To finish the proof we need to show that, with the assumption \begin{equation}\label{eq:ass-step3} \langle \xi_N, (\cH_N +\cN + \cN^2 / N) \xi_N\rangle \leq C \, ,  \end{equation} we have  
\begin{equation}\label{eq:step3-cl}\langle \Phi_{N,0}, (\cH_N +\cN) \Phi_{N,0} \rangle \leq C N^{\beta} .
\end{equation}
To reach this goal, we observe, first of all, that 
\begin{equation}\label{eq:phi0-1}
\begin{split} \langle \Phi_{N,0}, (\cH_N +\cN) \Phi_{N,0} \rangle &= \langle \1^{\le N} T_{N,0}^* \xi_N, (\cH_N +\cN) \1^{\le N} T_{N,0}^* \xi_N \rangle  \\ &\leq \langle \xi_N, T_{N,0} \cH_N T^*_{N,0} \xi_N \rangle +C
\end{split} 
\end{equation}
by Proposition \ref{lem:Bog-N} and (\ref{eq:ass-step3}). To bound the remaining expectation on the r.h.s. of (\ref{eq:phi0-1}), we compute (see \cite[Section 3, in particular Prop. 3.3 and Prop. 3.11]{BCS}) 
\begin{equation}\label{eq:THT-bd} \begin{split} 
T_{N,0}& \cH_N T_{N,0}^* \\ = \; &\cH_N + \|\nabla_2 \sinh_{k_{N,0}} \|^2  + N \int dxdy\;\big[\Delta \omega_N (x-y) \varphi^2_0 ((x+y)/2) \, a^*_x a^*_y+\text{h.c.} \big]\\
& + \frac{1}{2N}\int dx dy\; V_N(x-y)|\langle \text{sh}_x - \ph_0 (x) \text{sh}_{k_{N,0}} (\ph_0), \text{ch}_y - \ph_0 (y) \text{ch}_{k_{N,0}} (\ph_0) \rangle |^2 \\
        &+ \frac{1}{2}\int dx dy\; V_N(x-y) \big[-\omega_N(x-y)\varphi^2_{0}((x+y)/2)a_x^*a_y^*+\text{h.c.}\big] + \delta_N.
\end{split} \end{equation}
where we used the notation $\text{sh}_x$ to indicate the function $\text{sh}_x (z) = \sinh_{k_{N,0}} (x;z)$ and similarly for $\text{ch}_x$ (in this case, a distribution) and where the operator $\delta_N$ is such that 
\[ \pm \delta_N \leq \cH_N + C (\cN+\cN^2/N + 1)  \]
(in fact, the constant in front of $\cH_N$ could be chosen arbitrarily small, but we are not going to use this fact here). With (\ref{eq:kNt-bd}), we find
\[ \|\nabla_2 \sinh_{k_{N,0}} \|^2 \leq C N^{\beta} \]
Furthermore, integrating by parts, using \eqref{eq:fN-bound}, the assumption $\ph_0 \in H^4 (\bR^3)$ and (\ref{eq:ass-step3}), we obtain 
\begin{equation*}
\begin{aligned}
\Big|N \int dxdy\; \Delta & \omega_N (x-y) \varphi^2_{0}((x+y)/2)\langle \xi_N, a^*_xa^*_y \xi_N \rangle\Big| 
\\ &\leq \int dx \, \|a_x\xi_N\| \, \|a^* (N \nabla \omega_N (x-\cdot) \nabla_x \ph_0^2 ((x+\cdot)/2)) \xi_N\| \\
&\hspace{.3cm} + \int dx \, \|\nabla_x a_x \xi_N \| \, \|a^*( N \nabla \omega_N (x-\cdot) \varphi^2_{0} ((x+\cdot)/2)) \xi_N\| \\ 
&\leq \| (\cN+1)^{1/2} \xi_N \| \int dx \, \| a_x \xi_N \| \| N \nabla \omega_N (x-\cdot) \nabla_x \ph_0^2 ((x+\cdot)/2)) \|_2  \\ &\hspace{.3cm} + \| (\cN+1)^{1/2} \xi_N \| \int dx \, \| \nabla_x a_x \| \| N \nabla \omega_N (x-\cdot) \ph_0^2 ((x+\cdot)/2)) \|_2  \\
&\leq C N^\beta \| (\cN + \cK + 1)^{1/2} \xi_N \|^2 \leq C N^\beta.
\end{aligned}
\end{equation*}

Let us now consider the fourth term on the r.h.s. of (\ref{eq:THT-bd}). The most singular contribution is bounded by 
\[ \begin{split} \frac{1}{2N} \int &dx dy \, V_N (x-y) |\langle \text{sh}_x , \text{ch}_y \rangle |^2 \\ \leq \; & \frac{1}{2N} \int dx dy \, V_N (x-y) |\text{sh}_{k_{N,0}} (x;y)|^2 \\ &+ \frac{1}{2N} \int dx dy \, V_N (x-y) \Big| \int dz \, \text{sh}_{k_{N,0}} (x;z) \text{p} (y;z) \Big|^2 \\ \leq \; &N^{\beta-1} \big( \| \nabla_1 \text{sh}_{k_{N,0}} \|^2 + \| \nabla_2 \text{sh}_{k_{N,0}} \|^2 \big) + \frac{1}{2N} \int dx dy \, V_N (x-y) \| \text{sh}_x \|^2 \| \text{p}_y \|^2  \\ \leq \; &C N^{2\beta -1}  \end{split} \]
where we used Cauchy-Schwarz and the operator inequality 
$$V_N(x-y)\leq CN^{\beta}(-\Delta_x-\Delta_y)\, .$$

Finally, let us consider the fifth term on the r.h.s. of (\ref{eq:THT-bd}). Using Cauchy-Schwarz, \eqref{eq:fN-bound} and (\ref{eq:phi-bds}), we find 
\begin{equation*}
\begin{aligned}
 &\left|\int dx dy\; V_N(x-y) \, \omega_N(x-y)\varphi^2_0 ((x+y)/2)\langle \xi, a_x^*a_y^* \xi \rangle\right|  \\ & \qquad \leq  C \langle \xi, \cV_N \xi\rangle
   +C \int dx dy\; V_N(x-y) N |\omega_N(x-y)|^2 \, |\varphi_0 ((x+y)/2)|^4 
   \\ & \qquad \leq  C \delta \langle \xi, \cV_N \xi\rangle + CN^{2\beta-1} 
\end{aligned}
\end{equation*}

From (\ref{eq:THT-bd}), we conclude with (\ref{eq:ass-step3}) that
\[ \langle \xi, T_{N,0} \cH_N T_{N,0}^* \xi \rangle \leq C N^\beta \]
Together with (\ref{eq:phi0-1}), this implies (\ref{eq:step3-cl}). 
\end{proof}

A bound similar to the one in Lemma \ref{lem:Phi} also holds for the modified evolution $\w\Phi_{N,t}$ introduced in (\ref{eq:wPhi}).
\begin{lemma} \label{lem:wPhi} 
Assume Hypothesis A holds true. Let $\xi_N \in \cF_{\perp \ph_0}$ with $\| \xi_N \| \leq 1$ and 
\[ \langle \xi_N, (\cH_N +\cN + \cN (\cN/N)^{2b} ) \xi_N\rangle \leq C, \] 
uniformly in $N$. Let $\w\Phi_{N,t}$ be as defined in \eqref{eq:wPhi}. We assume here that the parameter $C_b > 0$ in (\ref{eq:wLNt}) is large enough. Then there exists a constant $C > 0$ such that 
\begin{equation}\label{eq:lemwphi}
\langle \w\Phi_{N,t}, (\cH_N +\cN+\cN(\cN/N)^{2b}) \w\Phi_{N,t}  \rangle \le  C N^{\beta} \exp(C\exp(C|t|)).
\end{equation}
for all $t \in \bR$. 
\end{lemma} 

\begin{proof}
Consider the Bogoliubov transformed dynamics $\xi_{N,t}= T_{N,t}\w\Phi_{N,t}$ as defined in \eqref{eq:defxiNt}. Then 
\[ \begin{split} \langle \wt{\Phi}_{N,t} , (\cH_N +\cN+\cN(\cN/N)^{2b}) \w\Phi_{N,t} \rangle &= \langle \xi_{N,t}, T_{N,t} (\cH_N +\cN+\cN(\cN/N)^{2b}) T_{N,t}^* \xi_{N,t} \rangle \\ &\leq C N^\beta \langle \xi_{N,t}, (\cH_N  + \cN + \cN (\cN/N)^{2b}) \xi_{N,t} \rangle  \end{split} \]
where we proceeded exactly as in Step 3 in the proof of Lemma \ref{lem:Phi} to bound the expectation of $T_{N,t} \cH_N T_{N,t}^*$ and we applied Proposition \ref{lem:Bog-N} to bound the other terms. Now we apply Corollary \ref{lem:xiN} to conclude that, if $\ell > 0$ is small enough in (\ref{eq:Neum}) and if $C_b>0$ is large enough in (\ref{eq:wLNt}), there exists a constant $C > 0$ such that 
\[ \langle \wt{\Phi}_{N,t} , (\cH_N +\cN+\cN(\cN/N)^{2b}) \w\Phi_{N,t} \rangle \leq C N^\beta \exp(C\exp(C|t|)) \]
for all $t \in \bR$. 
\end{proof}

Remark that Corollary \ref{lem:xiN} and Proposition \ref{lem:Bog-N} actually imply the stronger (compared with (\ref{eq:lemwphi})) estimate $\langle \w\Phi_{N,t} , \cN \w\Phi_{N,t} \rangle \leq C \exp (C \exp (C|t|))$ for the expectation of $\cN$. 

Using Lemma \ref{lem:Phi} and Lemma \ref{lem:wPhi} we are now ready to prove Lemma \ref{lem:Phi-wPhi}.
\begin{proof}[Proof of Lemma \ref{lem:Phi-wPhi}] 
Note that
$$
\| \Phi_{N,t}-\w\Phi_{N,t}\|^2 = 2 \big( 1-\Re \langle \Phi_{N,t},\w\Phi_{N,t}\rangle \big).
$$

With the notation $\1^{\le m}=\1(\cN\le m)$ and $\1^{>m}=\1-\1^{\le m}$, we can decompose 
\bq
\langle \Phi_{N,t},  \widetilde\Phi_{N,t} \rangle = \langle \Phi_{N,t}, \1^{\le m} \widetilde\Phi_{N,t} \rangle +  \langle \Phi_{N,t}, \1^{> m} \widetilde\Phi_{N,t} \rangle .
\eq
Instead of fixing $m$, we take the average over $m\in [M/2+1,M]$ with an even number $1\ll M \ll N$. This gives
\bq \label{eq:loc-dec}
\langle \Phi_{N,t},  \widetilde\Phi_{N,t} \rangle = \frac{2}{M}\sum_{m=M/2+1}^M \Big( \langle \Phi_{N,t}, \1^{\le m} \widetilde\Phi_{N,t} \rangle +  \langle \Phi_{N,t}, \1^{> m} \widetilde\Phi_{N,t} \rangle \Big).
\eq
We are going to choose $M= N^{1-\eps}$ with $\eps>0$ a sufficiently  small that will be specified later. Next, we estimate the two terms on the r.h.s. of \eqref{eq:loc-dec}. 

\bigskip

\noindent
{\it  Many-particle sectors.} With $\1^{> m} \le \cN/m$ 
and Lemma \ref{lem:wPhi}, we have
\[ \begin{split} 
|\langle \Phi_{N,t}, \1^{> m} \widetilde\Phi_{N,t} \rangle| &\le \| \Phi_{N,t}\| \|\1^{> m} \widetilde\Phi_{N,t}\| \\ &\le \langle \widetilde\Phi_{N,t}, (\cN/m) \widetilde\Phi_{N,t} \rangle^{1/2} \le C \sqrt{\frac{N^{\beta}}{M}}\exp(C\exp(C|t|)).
\end{split} \]
Thus 
\bq \label{eq:many-p}
\frac{2}{M} \sum_{m=M/2+1}^M |\langle \Phi_{N,t}, \1^{> m} \widetilde\Phi_{N,t} \rangle| \le C \sqrt{\frac{N^{\beta}}{M}}\exp(C\exp(C|t|)).
\eq

\noindent
{\it Few-particle sectors.} From the Schr\"odinger equations (\ref{eq:Phi}) and (\ref{eq:wPhi}) for $\Phi_{N,t}$ and $\widetilde\Phi_{N,t}$, we obtain  
\begin{align*}
\frac{d}{dt} \Re \langle \Phi_{N,t}, \1^{\le m} \widetilde\Phi_{N,t} \rangle = \Im \langle \Phi_{N,t}, (\cL_{N,t} \1^{\le m} - \1^{\le m} \widetilde\cL_{N,t} ) \widetilde\Phi_{N,t} \rangle 
\end{align*}
We can write
\begin{align*}
\cL_{N,t} \1^{\le m} - \1^{\le m} \widetilde\cL_{N,t} &= (\cL_{N,t}-\w\cL_{N,t})\1^{\le m}+ [ \w\cL_{N,t}, \1^{\le m}].   
\end{align*}

\noindent 
{\it Bound for $(\cL_{N,t} -\w\cL_{N,t}) \1^{\le m}$}. We have
\begin{equation}\label{eq:LwLm}
\begin{split}
 (\cL_{N,t} -\w\cL_{N,t})\1^{\le m} = \; & A_1 \Big[ \sqrt{1-\cN/N} - G_b(\cN/N)\Big] \1^{\le m}+{\rm h.c.} \\
&+ A_2 \frac{\sqrt{(N-\cN)(N-\cN-1)}-(N-\cN)}{N} \1^{\le m}+ {\rm h.c.} \\
&- C_b e^{C_bt} \cN (\cN/N)^{2b} \1^{\le m}
\end{split}
\end{equation}
with the two operators 
\begin{equation}\label{eq:A1A2}
\begin{split} 
A_1 = \; &\sqrt{N} \big[ a^* (Q_{N,t}[(V_N \omega_N) \ast|\varphi_{N,t}|^2] \varphi_{N,t}) - a^* (Q_{N,t}[V_N\ast|\varphi_{N,t}|^2] \varphi_{N,t}) (\cN/N) \big]  \\
&+ \frac{1}{\sqrt{N}}\int dx dy dx' dy'  \, (Q_{N,t} \otimes Q_{N,t} V_N Q_{N,t} \otimes 1)(x,y; x' , y') a_x^* a_y^* a_{x'} \varphi_{N,t} (y') 
\\
A_2 = \; &\frac{1}{2} \int dx dy \, K_{2,N,t}(x;y) a^*_x a^*_y . 
\end{split}
\end{equation}
To bound the r.h.s. of (\ref{eq:LwLm}) we are going to use the following proposition. 
\begin{proposition}\label{lem:A1} 
Assume the interaction potential $V$ to be smooth, spherically symmetric, compactly supported and non-negative. Then, for all 
vectors $\xi_1,\xi_2 \in \cF_{\bot \varphi_{N,t}}$, we have the bounds 
$$
|\langle \xi_1, A_{1} \xi_2\rangle| \le C \exp(C|t|) \big\langle \xi_1, \big(N^{2\beta-1} + (\cN/N)^2 + \cV_N \big) \xi_1\big\rangle^{1/2} \, \langle \xi_2, (\cN+1) \xi_2\rangle^{1/2} 
$$
and
$$
|\langle \xi_1, A_{2} \xi_2\rangle| \le C \sqrt{N} \exp(C|t|) \langle \xi_1,\cV_N \xi_1 \rangle^{1/2} \|\xi_2\|.
$$
\end{proposition}
\begin{proof} First we consider $A_1$.  Using 
$$a^*(g)a(g)\le a(g)a^*(g)\le (\cN+1) \|g\|_{L^2}^2$$
and 
\begin{align*}
\|Q_{N,t}[(V_N \omega_N) \ast|\varphi_{N,t}|^2]\varphi_{N,t}\|_{L^2} \le\; & \|[(V_N \omega_N) \ast|\varphi_{N,t}|^2]\varphi_{N,t}\|_{L^2} \\
\le \; & \|V_N\|_{L^1} \|\omega_N\|_{L^\infty} \|\varphi_{N,t}\|_{L^\infty}^2 \|\varphi_{N,t}\|_{L^2} \\ \le\; &  C N^{\beta-1} \exp(C|t|),
\end{align*}
we have
\begin{align*}
&|\langle \xi_1, \sqrt{N}a^* (Q_{N,t}[(V_N \omega_N) \ast|\varphi_{N,t}|^2] \varphi_{N,t}) \xi_2 \rangle| \\
&\hspace{2cm} \le C N^{\beta-1/2} \exp(C|t|) \|\xi_1\| \langle \xi_2, (\cN+1) \xi_2\rangle^{1/2} 
\end{align*}
and
\begin{align*}
|\langle \xi_1, a^* (Q_{N,t}[V_N\ast|\varphi_{N,t}|^2] \varphi_{N,t}) \frac{\cN}{N}  \xi_2 \rangle| \le C \exp(C|t|) \langle \xi_1, (\cN/N)^2 \xi_1 \rangle^{1/2} \langle \xi_2, \cN \xi_2\rangle^{1/2}. 
\end{align*}
Moreover\footnote{Note that the projection $Q_{N,t}$ has no effect in the excited Fock space $\cF_{\bot \varphi_{N,t}}$}
\begin{align*}
&\Big| \Big\langle \xi_1,  \frac{1}{\sqrt{N}}\int \d x \d y \d x' \d y' \, (Q_{N,t} \otimes Q_{N,t} V_N Q_{N,t} \otimes 1)(x,y; x',y') a_x^* a_y^* a_{x'} \varphi_{N,t} (y') \xi_2 \Big \rangle \Big| \\
&=  \Big| \frac{1}{\sqrt{N}} \int \d x \d y \, V_N(x-y) \varphi_{N,t} (y) \langle a_x a_y \xi_1, a_x \xi_2 \rangle  \Big| \\
&\le \frac{1}{\sqrt{N}} \int \d x \d y \, V_N(x-y) |\varphi_{N,t} (y)| \|a_x a_y \xi_1\| \|a_x \xi_2\| \\
&\le \|\varphi_{N,t}\|_{L^\infty} \Big( \frac{1}{N} \int \d x \d y \, V_N(x-y) \|a_x a_y \xi_1\|^2 \Big)^{1/2} \Big( \int \d x \d y V_N(x-y)  \|a_x \xi_2\|^2\Big)^{1/2}\\
&\le C\exp(C|t|) \langle \xi_1, \cV_N \xi_1 \rangle^{1/2} \langle \xi_2, \cN \xi_2 \rangle^{1/2}.
\end{align*}

To prove the bound for $A_2$, we estimate
\begin{align*}
|\langle \xi_1, A_2 \xi_2 \rangle| &= \Big| \int \d x \d y V_N(x-y) \varphi_{N,t}(x) \varphi_{N,t}(y)  \langle a_x a_y \xi_1, \xi_2 \rangle  \Big| \\
&\le \|\varphi_{N,t}\|_{L^\infty} \Big( \int \d x \d y \, V_N(x-y) \|a_x a_y \xi_1\|^2 \Big)^{1/2} \\
&\hspace{3cm} \times \Big( \int \d x \d y V_N(x-y) |\varphi_{N,t}(x)|^2  \|\xi_2\|^2\Big)^{1/2}\\
&\le C \sqrt{N} \exp(C|t|)\langle \xi_1, \cV_N \xi_1\rangle^{1/2} \|\xi_2\|.
\end{align*}
This ends the proof of the proposition.
\end{proof}

We control now the operators on the r.h.s. of (\ref{eq:LwLm}).  Obviously,
$$
\cN(\cN/N)^{2b}\1^{\le m} \le CM(M/N)^{2b}.
$$
and, therefore,
$$
|\langle \Phi_{N,t}, \cN(\cN/N)^{2b}\1^{\le m} \w\Phi_{N,t} \rangle | \le CM(M/N)^{2b}.
$$

Using Proposition \ref{lem:A1} with 
$$\xi_1=\Phi_{N,t}, \quad \xi_2= \big[ \sqrt{1-\cN/N}-G_b(\cN/N) \big] \1^{\le m} \, \w\Phi_{N,t},$$
combined with the simple bound 
$$
|\sqrt{1-\cN/N}-G_b(\cN/N)|\1^{\le m} \le C(M/N)^{b+1}
$$
that follows from \eqref{eq:Pb-1} and with the estimates in Lemma \ref{lem:Phi} and Lemma \ref{lem:wPhi}, we obtain
\begin{align*}
\Big| \langle \Phi_{N,t}, A_1 (\sqrt{1-\cN/N}-G_b(\cN/N))\1^{\le m} \w\Phi_{N,t}\rangle\Big| \le C (M/N)^{b+1} N^{\beta} \exp(C\exp(C|t|)).
\end{align*}

Using again Proposition \ref{lem:A1} with
$$\xi_1=\Phi_{N,t}, \quad \xi_2= \big[ \sqrt{(N-\cN)(N-\cN-1)} - N-\cN \big] \1^{\le m} \w\Phi_{N,t} \, , $$
the simple bound  
$$
|\sqrt{(N-\cN)(N-\cN-1)} - N-\cN| \le 1,
$$ 
and the bounds in Lemma \ref{lem:Phi} and Lemma \ref{lem:wPhi}, 
we also obtain
\begin{align*}
\Big| \big\langle \Phi_{N,t}, A_2 \frac{\sqrt{(N-\cN)(N-\cN-1)} - N-\cN}{N} \1^{\le m} \w\Phi_{N,t}\big\rangle\Big| \le C N^{\frac{\beta-1}{2}} \exp(C\exp(C|t|)).
\end{align*}

The hermitian conjugated terms can be controlled analogously (Proposition \ref{lem:A1} provides bounds for $A_1^*, A_2^*$, as well, switching $\xi_1$ and $\xi_2$). In summary, we have shown that 
\begin{align*}
&\Big| \langle \Phi_{N,t},  (\cL_{N,t} -\w\cL_{N,t})\1^{\le m} \w\Phi_{N,t}\rangle\Big| \\
&\hspace{2cm} \le C \Big[ N^{\beta-1}+ M(M/N)^{2b} + (M/N)^{b+1} N^{\beta}\Big] \exp(C\exp(C|t|)).
\end{align*}

\noindent{\it Bound for $[ \w\cL_{N,t}, \1^{\le m}]$.} We can decompose
\bq
[\w\cL_{N,t},\1^{\le m}] = \1^{\le m} \widetilde\cL_{N,t} \1^{>m} - \1^{> m} \w\cL_{N,t} \1^{\le m}. 
\eq
Let us focus on $ \1^{> m} \w\cL_{N,t} \1^{\le m}$; the other term can be treated similarly. With the operators $A_1, A_2$ defined in (\ref{eq:A1A2}), we have 
\begin{equation}\label{eq:1L1}
\begin{split}
\1^{> m} \w\cL_{N,t} \1^{\le m} &= \1^{> m} \Big(A_1 G_p(\cN/N) + A_2\frac{N-\cN}{N} \Big) \1^{\le m} \\
&= A_1 G_p(\cN/N) \1(\cN=m)+ A_2\frac{N-\cN}{N} \1(m-1\le \cN \le m).
\end{split} 
\end{equation}
Here we used the fact that $A_1$ creates exactly one particle while  $A_2$ creates exactly two particles. All other terms in $\w\cL_{N,t}$ leave the number of particles invariant, and therefore do not contribute to (\ref{eq:1L1}). Thus
\begin{align*}
\sum_{m=M/2+1}^M \1^{> m} \w\cL_{N,t} \1^{\le m} = \; &A_1 G_p(\cN/N) \1(M/2<\cN \le M ) \\
& + A_2\frac{N-\cN}{N} \Big[ \1(M/2< \cN \le M) + \1(M/2\le \cN<M) \Big].
\end{align*}

Using Proposition \ref{lem:A1} with
$$\xi_1=\Phi_{N,t}, \quad \xi_2 = G_p (\cN/N) \1(M/2<\cN \le M ) \w\Phi_{N,t},$$
combined with the simple estimate (recall that we will choose $M \ll N$) 
$$
|G_p(\cN/N)| \1(M/2<\cN \le M) \le C
$$ 
and with the bounds in Lemma \ref{lem:Phi} and in Lemma 
\ref{lem:wPhi}, we obtain
\begin{align*}
\langle \Phi_{N,t},  A_1 G_p(\cN/N) \1(M/2<\cN \le M ) \w\Phi_{N,t}\rangle \le C N^{\beta} \exp(C\exp(C|t|)).
\end{align*}

Similarly, using again Proposition \ref{lem:A1} and Lemma \ref{lem:wPhi}, we find 
\begin{align*}
\langle \Phi_{N,t},  A_2 (1-\cN/N)\Big[ \1(M/2< \cN \le M) + \1(M/2 &\le \cN<M) \Big] \w\Phi_{N,t}\rangle \\
&\le C N^{\frac{\beta+1}{2}} \exp(C\exp(C|t|)).
\end{align*}
Thus, we conclude that 
\bq
\frac{2}{M} \Big| \sum_{m=M/2+1}^M \langle \Phi_{N,t},  [\w\cL_{N,t},\1^{\le m}]  \w\Phi_{N,t}\rangle \Big| \le C \frac{N^{\frac{\beta+1}{2}}}{M} \exp(C\exp(C|t|)). 
\eq
In summary, we have proved that 
\begin{align*}
&\Big| \Re \frac{2}{M} \sum_{m=M/2+1}^M \frac{d}{dt}  \langle \Phi_{N,t}, \1^{\le m} \widetilde\Phi_{N,t} \rangle\Big|\\
&\hspace{1cm} \le  C \Big[ N^{\frac{\beta-1}{2}}+ M \Big(\frac{M}{N}\Big)^{2b} + N^{\beta} \Big(\frac{M}{N}\Big)^{b+1} +  
\frac{N^{\frac{\beta+1}{2}}}{M}\Big] \exp(C\exp(C|t|)).
\end{align*}

\noindent {\it Conclusion of the proof.} For every $\alpha<(1-\beta)/2$, we can choose $M=N^{1-\eps}$ with a sufficiently small $\eps>0$, and then $b$ sufficiently large to obtain 
\begin{align*}
\Big| \Re \frac{2}{M} \sum_{m=M/2+1}^M \frac{d}{dt}  \langle \Phi_{N,t}, \1^{\le m} \widetilde\Phi_{N,t} \rangle\Big| \le  C N^{-\alpha} \exp(C\exp(C|t|)).
\end{align*}
Integrating over $t$, we find 
\[ \begin{split}
\Re \frac{2}{M} \sum_{m=M/2+1}^M  &\langle \Phi_{N,t}, \1^{\le m} \widetilde\Phi_{N,t} \rangle \\ & \ge \Re \frac{2}{M} \sum_{m=M/2+1}^M  \langle \Phi_{N,0}, \1^{\le m} \widetilde\Phi_{N,0} \rangle - C N^{-\alpha} \exp(C\exp(C|t|)).
\end{split}\]
On the other hand, using the assumption $\Phi_{N,0} =\1^{\le N}T^*_{N,0}\xi_N$, $\w\Phi_{N,0}=T^*_{N,0}\xi_N$ we have the lower bound 
\begin{align*}
 \langle \Phi_{N,0}, \1^{\le m} \widetilde\Phi_{N,0} \rangle &= \| \1^{\le m}T_{N,0}^*\xi_{N}\|^2 = 1 - \| \1^{> m}T^*_{N,0}\xi_{N}\|^2 \\
 &\ge 1- \langle T^*_{N,0}\xi_N, (\cN/m) T^*_{N,0}\xi_N\rangle \\
 & \ge 1 - C \langle \xi_N, (\cN/m) \xi_N\rangle\ge 1-C/M.
\end{align*}
Here we have used Proposition \ref{lem:Bog-N} in the second last estimate and the assumption on $\xi_N$ for the last inequality. Thus
\begin{align*}
\Re \frac{2}{M} \sum_{m=M/2+1}^M  \langle \Phi_{N,t}, \1^{\le m} \widetilde\Phi_{N,t} \rangle \ge 1 - C N^{-\alpha} \exp(C\exp(C|t|))- C M^{-1}
\end{align*}
Combining with \eqref{eq:many-p} and using the choice $M = N^{1-\eps}$ for a sufficiently small $\eps > 0$, we obtain 
\begin{align*}
\Re \langle \Phi_{N,t},  \widetilde\Phi_{N,t} \rangle  \ge 1 - C N^{-\alpha} \exp(C\exp(C|t|)).
\end{align*}
Consequently,
\begin{align*}
\|\Phi_{N,t}-  \widetilde\Phi_{N,t} \|^2 \le 2(1-\Re \langle \Phi_{N,t},  \widetilde\Phi_{N,t}\rangle) \le C N^{-\alpha} \exp(C\exp(C|t|)).
\end{align*}
\end{proof}

\section{Proof of main Results}
\label{sec:thm-proof}

Combining Lemma \ref{lem:Phi-wPhi} and Lemma \ref{lem:Compare-xiN}, we can prove our first main theorem. 

\begin{proof}[Proof of Theorem \ref{thm:main1}]
Fix $\alpha < \min (\beta/2,(1-\beta)/2)$. To begin with, let us choose a sequence $\xi_N \in \cF_{\perp \ph_0}$ with $\| \xi_N \| \leq 1$ and with 
\begin{equation}\label{eq:ass-xi-str} \langle \xi_N , (\cH_N + \cN + \cN (\cN/N)^{2b}) \xi_N \rangle \leq C \end{equation}
uniformly in $N$. This assumption is stronger than the assumption (\ref{eq:ass-xi}) in the theorem; at the end, we will show how 
to relax it.

Assuming (\ref{eq:ass-xi-str}), we consider the many-body evolution 
\[ 
\Psi_{N,t} = e^{-itH_N} U_{\ph_0}^* \1^{\leq N} T_{N,0}^* \xi_N
\]
and we factor out the condensate, defining, as in (\ref{eq:phiNt}), 
$\Phi_{N,t} = U_{\ph_{N,t}} \Psi_{N,t}$. To prove Theorem \ref{thm:main1}, we have to compare $\Phi_{N,t}$ with the (Bogoliubov transformed) effective evolution $T^*_{N,t} \xi_{2,N,t} = T^*_{N,t} \cU_{2,N} (t;0) \xi_N$. To this end, we recall the definition (\ref{eq:wPhi}) of the modified fluctuation dynamics $\wt{\Phi}_{N,t}$, and we bound
\[ \left\| \Phi_{N,t} - T^*_{N,t} \xi_{2,N,t} \right\| \leq \| \Phi_{N,t} - \wt{\Phi}_{N,t} \| + \| \wt{\Phi}_{N,t} - T^*_{N,t} \xi_{2,N,t} \| \leq \| \Phi_{N,t} - \wt{\Phi}_{N,t} \| + \| \xi_{N,t} -  \xi_{2,N,t} \| \]
where, as in (\ref{eq:defxiNt}), we set $\xi_{N,t} = T_{N,t} \wt{\Phi}_{N,t}$ and we used the unitarity of $T_{N,t}$. Combining Lemma \ref{lem:Phi-wPhi} and Lemma \ref{lem:Compare-xiN} (which can be used, because of the additional assumption (\ref{eq:ass-xi-str})), we conclude that there exists a constant $C > 0$ such that 
\begin{equation}\label{eq:res-add} \left\| \Phi_{N,t} - T^*_{N,t} \xi_{2,N,t} \right\|_\cF \leq C N^{-\alpha} \exp (C \exp (C|t|)) \end{equation}
for all $t \in \bR$ and all $N$ large enough. This proves Theorem \ref{thm:main1} under the additional assumption (\ref{eq:ass-xi-str}). 

Now, let us assume that the sequence $\xi_N \in \cF_{\perp \ph_0}$ is normalized $\| \xi_N \| =1$, but, instead of (\ref{eq:ass-xi-str}), that it only satisfies the weaker bound 
\begin{equation}\label{eq:ass-thm} \langle \xi_N , (\cK + \cN) \xi_N \rangle \leq C \, , \end{equation}
uniformly in $N$. We choose $M = N^{2\alpha}$ and we decompose
\[ \xi_N = \1^{\leq M} \xi_N + \1^{>M} \xi_N \]
Then, using unitarity of the maps $U_{\ph_{N,t}}$, $T_{N,t}$, $e^{iH_N t}$ and $\cU_{2,N} (t;0)$, we obtain 
\begin{equation}\label{eq:fin} 
\begin{split}\| \Phi_{N,t} - T^*_{N,t} \xi_{2,N,t} \| = \; &\| U_{\ph_{N,t}} e^{-itH_N} U_{\ph_0}^* \1^{\leq N} T_{N,0}^* \xi_N - T^*_{N,t} \cU_{2,N} (t;0) \xi_N \| \\ \leq \; &\| U_{\ph_{N,t}} e^{-itH_N} U_{\ph_0}^* \1^{\leq N} T_{N,0}^* \1^{\leq M} \xi_N - T^*_{N,t} \cU_{2,N} (t;0) \1^{\leq M} \xi_N \| \\&+ 2 \| \1^{>M} \xi_N \| \end{split} \end{equation} 
On the one hand, using Markov's inequality and (\ref{eq:ass-thm}), we have
\[ \| \1^{> M} \xi_N \|^2 = \langle \xi_N, \1^{>M} \xi_N \rangle \leq M^{-1} \langle \xi_N, \cN \xi_N \rangle \leq C N^{-2\alpha} \]
On the other hand, the sequence $\wt{\xi}_N = \1^{\leq M} \xi_N$ is such that $\| \wt{\xi}_N \| \leq \| \xi_N \| = 1$ and 
\begin{equation}\label{eq:ass-thm2} \langle \wt{\xi}_N, (\cH_N + \cN + \cN (\cN/N)^{2b}) \wt{\xi}_N \rangle \leq \langle \xi_N, (\cK + \cN +1) \xi_N \rangle \leq C \end{equation}
by (\ref{eq:ass-thm}). Here we used the bound $\cV_N \leq C N^{\beta- 1} (\cK+1)(\cN +1)$ for the potential energy, which implies, by the choice of $M=N^{2\alpha}$ and of $\alpha \leq (1-\beta)/2$, that $\cV_N \1^{\leq M} \leq C (\cK+1)$. Because of (\ref{eq:ass-thm2}), we can apply the convergence (\ref{eq:res-add}), established under the additional assumption (\ref{eq:ass-xi-str}), to estimate the  first term on the r.h.s. of (\ref{eq:fin}). We obtain that (this time only under the assumption (\ref{eq:ass-thm}))   
\[ \| \Phi_{N,t} - T^*_{N,t} \xi_{2,N,t} \|\leq 
C N^{-\alpha} \exp (C \exp (C|t|)) \] 
This concludes the proof of Theorem \ref{thm:main1}.
\end{proof}

To show Theorem \ref{thm:main2}, we compare the difference between the generators of the quadratic evolutions $\cU_{2,N}$ and $\cU_2$ defined in (\ref{eq:U2Nts-def}) and, respectively, in (\ref{eq:U2infty}).
\begin{proposition}\label{lm:compare-xiinfty}
Assume Hypothesis A holds true. Let $\cG_{2,N,t}$ and $\cG_{2,t}$ be as defined in (\ref{eq:cG2N}) and in (\ref{eq:G2t}) (and $\eta_N (t)$ as in (\ref{eq:etaN})). Then there exists $C > 0$ such that, with $\alpha = \min (\beta/2 , (1-\beta)/2)$, 
\[ \begin{split} | \langle \xi_1, (\cG_{2,N,t} - \eta_N (t) &- \cG_{2,t}) \xi_2 \rangle | \\ &\leq C N^{-\alpha} \exp (C \exp (C |t|))  \| (\cK + \cN+1)^{1/2} \xi_1 \| \| (\cN + 1)^{1/2} \xi_2 \| \end{split} \]
for all $\xi_1, \xi_2 \in \cF_{\perp \ph_{N,t}}$ and all $t \in \bR$. 
\end{proposition} 

The proof of Proposition \ref{lm:compare-xiinfty} can be found in \cite[Lemmas 5.1, 5.2, 5.3, 5.4]{BCS}, up to very minor modifications. 

\begin{proof}[Proof of Theorem \ref{thm:main2}]
As in the proof of Theorem \ref{thm:main1}, we first assume that 
\begin{equation}\label{eq:ass-str-thm2} \langle \xi_N, (\cH_N + \cN + \cN (\cN/N)^{2b}) \xi_N \rangle \leq C \end{equation}
uniformly in $N$. With $ \theta_N (t):=-\int_0^t d\tau\;\eta_N (\tau)$ we find 
		\[\begin{split} &\frac{d}{dt} \big\| \xi_{2,N,t}  - e^{i\theta_N (t)} \xi_{2,t}  \big\|^2 = 2 \im \big\langle  \xi_{2,N,t}, \big[\cG_{2,N,t}-\eta_N(t)-\cG_{2,t}\big] e^{i\theta_N (t)} \xi_{2,t} \big\rangle 
        \end{split}\]
Proposition \ref{lm:compare-xiinfty} above implies that 
		\[\begin{split}
        \frac{d}{dt} \big\| \xi_{2,N,t} &- e^{i\theta_N (t)} \xi_{2,t}  \big\|^2 \\ &\leq CN^{-\alpha}\exp(C\exp(C|t|)) \langle \xi_{2,N,t} , (\cK+\cN+1) \xi_{2,N,t} \rangle^{1/2} \langle \xi_{2,t}, (\cN+1) \xi_{2,t} \rangle^{1/2} \\
        &\leq CN^{-\alpha} \exp(C\exp(C|t|))
        \end{split}\]
Here we used Corollary \ref{lem:xiN} (with the additional assumption (\ref{eq:ass-str-thm2})) and the analogous bound
\begin{equation}\label{eq:xi2-bd} \langle \xi_{2,t}, (\cN+1) \xi_{2,t} \rangle \leq C \exp (C \exp (C|t|)) \end{equation}
for the limiting dynamics $\xi_{2,t}$. Eq. (\ref{eq:xi2-bd}) can be proven similarly to the bound for $\xi_{2,N,t}$ in Corollary \ref{lem:xiN} (with estimates for the generator $\cG_{2,t}$  
analogous to (\ref{eq:G2Nt-est})). Integrating in time, we conclude that
\[ \big\| \xi_{2,N,t}  - e^{i\theta_N (t)} \xi_{2,t}  \big\|^2 \leq C N^{-\alpha} \exp ( C \exp (C |t|)) \]
for all $t \in \bR$. Combining the last bound with Theorem \ref{thm:main1}, we obtain 
\[ \begin{split}  \| U_{\ph_{N,t}} \Psi_{N,t} - e^{-i\theta_N (t)} T^*_{N,t} \xi_{2,t} \| &\leq \| U_{\ph_{N,t}} \Psi_{N,t} - T^*_{N,t} \xi_{2,N,t} \| +  \| \xi_{2,N,t} -  e^{-i\theta_N (t)} \xi_{2,t} \| 
\\ &\leq C N^{-\alpha/2} \exp (C \exp (C |t|)) \end{split} \]  
This proves Theorem \ref{thm:main2} under the additional assumption (\ref{eq:ass-str-thm2}). To relax this condition, we proceed exactly as in the proof of Theorem \ref{thm:main1}. We omit the details. 
\end{proof}

Finally, Theorem \ref{thm:main3} follows immediately combining Theorem \ref{thm:main2} with the following proposition, which is a modification of the analysis in \cite[Section 6]{BS}. 
\begin{proposition}\label{lm:init} 
Assume Hypothesis A holds true. Let $\psi_{N} \in L^2_s (\bR^{3N})$ with reduced one-particle density $\gamma_{N}$ such that
\begin{equation}\label{eq:aN} a_N := \tr \, \left| \gamma_{N} - |\ph_0 \rangle \langle \ph_0| \right| \leq C N^{-1} \end{equation}
and 
\begin{equation}\label{eq:bN} b_N := \left| \frac{1}{N} \langle \psi_{N} , H_N \psi_{N} \rangle - \big[ \| \nabla \ph_0 \|_2^2 + \frac{1}{2} \langle \ph_0, [ V_N f_N * |\ph_0|^2] \ph_0 \rangle \big] \right| \leq C N^{-1}   \end{equation}
Set $\xi_N = T_{N,0} U_{\ph_0} \psi_{N}$ with the Bogoliubov transformation $T_{N,0}$ defined in (\ref{eq:Bog-trans}). Then, we have $\psi_{N} = U_{\ph_0}^* {\bf 1}^{\leq N} T^*_{N,0} \xi_N$ and 
\[ \langle \xi_N, \left[ \cK + \cN \right] \xi_N \rangle \leq C \]
uniformly in $N$.
\end{proposition} 

\begin{proof}  
First of all, we remark that, with Proposition \ref{lem:Bog-N} and  (\ref{eq:U-rules}), 
\[ \begin{split} \langle \xi_N, \cN \xi_N \rangle &= \langle T_{N,0} U_{\ph_0} \psi_N, \cN  T_{N,0} U_{\ph_0} \psi_N \rangle \\ &\leq C \langle  U_{\ph_0} \psi_N, (\cN+1) U_{\ph_0} \psi_N \rangle \\ &= C \left[ N- \langle \psi_N, a^* (\ph_0) a(\ph_0) \psi_N \rangle \right] + C \\ &= C N \left[ 1 - \langle \ph_0, \gamma_N \ph_0 \rangle \right] + C  \leq C N a_N + C.  \end{split} \]
To bound $\langle \xi_N, \cK \xi_N \rangle$, we use $\cK \leq \cH_N$ and the first bound in (\ref{eq:cGmm}), which implies that  
\[ \begin{split} \langle \xi_N, \cH_N \xi_N \rangle &\leq 2 \langle \xi_N, (\cG_{N,0} - \eta_N (0))\xi_N \rangle + C \langle \xi_N , (\cN+1) \xi_N \rangle \\ &\leq 2 \langle \xi_N, (\cG_{N,0} - \eta_N (0))\xi_N \rangle + C N a_N +C. \end{split} \]
Hence, the proposition follows from (\ref{eq:aN}) and (\ref{eq:bN}) if we can show that 
\begin{equation}\label{eq:iniest} 
\langle \xi_N, \big[\cG_{N,0}-\eta_{N}(0) \big] \xi_N \rangle \leq \frac14\langle \xi_N, \cH_N \xi_N \rangle  + CN (a_N+b_N)+ C.  \end{equation}

To prove (\ref{eq:iniest}) we observe that, from the definition (\ref{def:wcG}) of $\cG_{N,0}$ and since $\xi_N =  T_{N,0} U_{\ph_0} \psi_{N}$, 
\begin{equation}\label{eq:G-eta1} \begin{split} \langle \xi_N, \big[\cG_{N,0}-\eta_N(0)\big]\xi_N \rangle  =& \langle U_{\ph_0} \psi_N, \big[ T_{N,0}^*(i\partial_t T_{N,t})_{|t=0}  +{\cL}_{N,0} -\eta_N(0) \big] U_{\ph_0}\psi_N\rangle \\
        &+  \langle U_{\ph_0}\psi_N, \big[\widetilde{\cL}_{N,0}-{\cL}_{N,0}\big]  U_{\ph_0}\psi_N\rangle. 
        \end{split}\end{equation}
{F}rom the proof of Lemma 6.2 and of Theorem 1.1 in \cite[Section 6]{BS}, we find 
        \begin{equation}\label{eq:inibnd1}\begin{split}
        \langle U_{\ph_0}\psi_N, \big[T_{N,0}^*(i\partial_t T_{N,t})_{|t=0}  +{\cL}_{N,0} -\eta_N(0)\big]U_{\ph_0}\psi_N\rangle \leq CN (a_N+b_N) + C. 
        \end{split}\end{equation}
Therefore, it is enough to consider the second term on the r.h.s. of (\ref{eq:G-eta1}). From the definitions (\ref{eq:wLNt}) of $ \widetilde{\cL}_{N,0}$ and (\ref{eq:cLNt-2}) of $\cL_{N,0}$, we have (see also (\ref{eq:LwLm})) \[ \wt{\cL}_{N,0} - \cL_{N,0} = \sum_{j=1}^4 D_j,\] with the operators
\[ \begin{split} 
D_1 = \; &\sqrt{N} \Big[a^*\big(Q_{N,0}\big[(V_N \omega_N)\ast|\varphi_0|^2\big]\varphi_0 \big) - a^* (Q_{N,0} [V_N *|\ph_{0}|^2] \ph_{0}) (\cN/N)  \Big] \\ &\times (G_b (\cN/N)-\sqrt{1-\cN/N}) + \text{h.c.} \\
D_2 = \; &\frac{1}{2} \int dxdy K_{2,N,0} (x;y) a_x^* a_y^* \frac{(N-\cN) - \sqrt{(N-\cN)(N-1-\cN)}}{N} + \text{h.c.} \\
D_3 = \; &\frac{1}{\sqrt{N}} \int dx dy (Q_{N,0} \otimes Q_{N,0} V_N Q_{N,0} \otimes 1) (x,y ; x' ,y') a_x^* a_y^* a_{x'} \ph_0 (y')  \\ &\times (G_b (\cN/N)-\sqrt{1-\cN/N}) + \text{h.c.} \\
D_4 = \; &C_b \cN ( \cN / N)^{2b}. \end{split} \]
Using $|\sqrt{1-z} - G_b (z)| \leq C z^{b+1}$ for all $z > 0$, we easily arrive at 
\begin{equation}\label{eq:D1-init} \big| \langle U_{\ph_0} \psi_N, D_1 U_{\ph_0} \psi_N \rangle \big| \leq C \langle U_{\ph_0} \psi_N, \cN U_{\ph_0} \psi_N \rangle \leq CNa_N + C. \end{equation} 
Since, for $z \in (0,1)$,   
\[ \big| (1-z) - \sqrt{(1-z)(1-z-1/N)} \big| \leq C/N \]
we obtain that, for any $\delta > 0$ (recall that $ Q_{N,0}$ has no effect on states in $\cF^{\leq N}_{\bot \varphi_0}$), 
\[ \begin{split} |\langle U_{\ph_0} &\psi_N , D_2 U_{\ph_0} \psi_N \rangle | \\  &\leq  \int dxdy\; N^{3\beta-1} V(N^\beta(x-y))\big(\delta^{-1}N |\varphi_0(x)|^2|\varphi_0(y)|^2 +  \delta N^{-1}\|a_xa_y U_{\ph_0}\psi_N\|^2\big)\\
        &\leq \delta N^{-1}\langle U_{\ph_0}\psi_N, \cV_N U_{\ph_0}\psi_N\rangle + C.
        \end{split}\]
As in Step 3 of the proof of Lemma \ref{lem:Phi}, we can estimate 
        \begin{equation}\label{eq:inibnd2}\begin{split}
        \delta N^{-1}\langle U_{\ph_0}\psi_N, \cV_N U_{\ph_0}\psi_N\rangle = &\; \delta N^{-1}\langle \xi_N, T_{N,0}\cH_N T_{N,0}^* \xi_N\rangle\\
        \leq&\; \delta \langle \xi_N, \cH_N  \xi_N\rangle + CN a_N + C.
        \end{split}\end{equation}
Choosing, for example, $\delta = 1/8$, we conclude that
\begin{equation}\label{eq:D2-init} |\langle U_{\ph_0} \psi_N , D_2 U_{\ph_0} \psi_N \rangle |  \leq  \frac{1}{8} \langle \xi_N, \cH_N  \xi_N\rangle + CN a_N + C. \end{equation}

As for the expectation of $D_3$, we proceed similarly as in the proof of Proposition \ref{lem:wG} (in particular, in the bound for the operator $B_2$). Using again the bound $|\sqrt{1-z}-G_b(z)| \leq C z^{b+1}$ for all $z \in (0;1)$,  we find that, for every $\delta > 0$ there exists $C > 0$ such that 
\begin{equation*}
        \begin{split}
        \big| &\langle U_{\ph_0} \psi_N, D_3 U_{\ph_0} \psi_N \rangle \big| \\
        &= \frac{1}{\sqrt N} \bigg|\int dx dy\; V_N(x-y)\varphi_0 (y) \langle\xi_N, T_{N,0} \, a_x^*a_y^*a_x\big(G_b(\cN/N)-\sqrt{1-\cN/N}\big)T_{N,0}^* \xi_N\rangle \bigg|\\
        &\leq \delta \langle \xi_N, \cV_N  \xi_N\rangle +  C \langle \xi_N, (\cN+1) \xi_N\rangle. 
        \end{split}        
        \end{equation*}
Choosing $\delta = 1/8$, we obtain
\begin{equation}\label{eq:inibnd3} \big| \langle U_{\ph_0} \psi_N, D_3 U_{\ph_0} \psi_N \rangle \big| \leq \frac{1}{8} \langle \xi_N, \cH_N \xi_N \rangle + C N a_N + C. \end{equation}

Finally, since $U_{\ph_0} \psi_N$ has at most $N$ particles, 
we easily find that
\[ 0 \leq \langle  U_{\ph_0} \psi_N, D_4 U_{\ph_0} \psi_N \rangle 
\leq C N a_N + C. \]

Combining the last bound with (\ref{eq:D1-init}), (\ref{eq:D2-init}) and (\ref{eq:inibnd3}), we conclude that
\[ | \langle  U_{\ph_0} \psi_N, \big[ \wt{\cL}_{N,0} - \cL_{N,0} \big]  U_{\ph_0} \psi_N \rangle | \leq \frac{1}{4} \langle \xi_N, \cH_N \xi_N \rangle + C N a_N + C \]
Together with (\ref{eq:inibnd1}) and (\ref{eq:G-eta1}), we obtain (\ref{eq:iniest}).
\end{proof}

\end{document}